\documentclass[11pt]{amsart}
\usepackage{graphicx}
\usepackage{amscd}
\usepackage{amsmath}
\usepackage{amsfonts}
\usepackage{amssymb}
\usepackage{bbm}
\usepackage{setspace}
\usepackage{enumerate}         % better lists
\usepackage{color}
\usepackage{url}
\usepackage{amsthm}
\usepackage{hyperref}
\usepackage{bm}
\usepackage{xy}
\usepackage{color}
\usepackage{xcolor}
\usepackage{subfig}
\usepackage{bbm}
\usepackage{graphicx}
\usepackage{etoolbox}
\usepackage[pagewise]{lineno}
\usepackage{mathrsfs}
\usepackage[normalem]{ulem}
\usepackage{eurosym}
\usepackage{cancel} 
\usepackage{soul}
\allowdisplaybreaks[4]

\usepackage{natbib}

\usepackage{geometry}
\usepackage{verbatim}

\usepackage{tikz}
\usetikzlibrary{positioning,shapes,shadows,arrows}
\tikzstyle{abstract}=[rectangle, draw=black, rounded corners, fill=blue!40, drop shadow,
        text centered, anchor=north, text=white, text width=2.5cm]
\tikzstyle{comment}=[rectangle, draw=black, rounded corners, fill=green, drop shadow,
        text centered, anchor=north, text=white, text width=2cm]
\tikzstyle{myarrow}=[->, >=open triangle 90, thick]
\tikzstyle{line}=[-, thick]

\theoremstyle{plain}
\newtheorem{theorem}{Theorem}[section]

\newtheorem{lemma}[theorem]{Lemma}
\newtheorem{notation}[theorem]{Notation}

\newtheorem{proposition}[theorem]{Proposition}

\newtheorem{definition}[theorem]{Definition}

\newtheorem{assumption}[theorem]{Assumption}
\newtheorem{assumptionStanding}[theorem]{Hypothesis}

\theoremstyle{remark}
\newtheorem{remark}[theorem]{Remark}

\numberwithin{equation}{section}

\newcommand{\rsto}{]\!\kern-1.8pt ]}
\newcommand{\lsto}{[\!\kern-1.7pt [}

\numberwithin{equation}{section}

\newcommand{\FF}{\mathbb{F}}
\newcommand{\GG}{\mathbb{G}}

\newcommand{\RR}{\mathbb{R}}

\newcommand{\PP}{\mathbb{P}}

\newcommand{\EE}{\mathbb{E}}

\newcommand{\cG}{\mathcal{G}}
\newcommand{\cH}{\mathcal{H}}

\newcommand{\cV}{\mathcal{V}}

\newcommand{\rck}[1]{r^{c,#1}} % collateral rate for currency i
\newcommand{\tI}{\mathcal{I}}% set of indices for cash account 

\usepackage{enumitem}
\newcounter{descriptcount}
\newlist{enumdescript}{description}{1}
\setlist[enumdescript,1]{%
  before={\setcounter{descriptcount}{0}%
          \renewcommand*\thedescriptcount{\arabic{descriptcount}}},
        font={\bfseries\stepcounter{descriptcount}Case \thedescriptcount~}
}

\renewcommand{\cite}{\citet}

\makeatletter
\patchcmd{\@maketitle}
  {\ifx\@empty\@dedicatory}
  {\ifx\@empty\@date \else {\vskip3ex \centering\footnotesize\@date\par\vskip1ex}\fi
   \ifx\@empty\@dedicatory}
  {}{}
\patchcmd{\@adminfootnotes}
  {\ifx\@empty\@date\else \@footnotetext{\@setdate}\fi}
  {}{}{}
\newcommand{\subjclassname@JEL}{JEL Classification}
\makeatother

\begin{document}

\title[Incompleteness and xVA]{
 When defaults cannot be hedged: XVA calculations  via Local Risk-Minimization}
%}}

\author{Francesca Biagini}
\address[Francesca Biagini]{LMU M\"unchen, Mathematics Institute,\newline
\indent Theresienstr. 39, D-80333 Munich, Germany}
\email{biagini@math.lmu.de}%

\author{Alessandro Gnoatto}
\address[Alessandro Gnoatto]{University of Verona, Department of Economics, \newline
\indent via Cantarane 24, 37129 Verona, Italy}
\email{alessandro.gnoatto@univr.it}

\author{Katharina Oberpriller}
\address[Katharina Oberpriller]{LMU M\"unchen, Mathematics Institute,\newline
\indent Theresienstr. 39, D-80333 Munich, Germany}
\email{oberpriller@math.lmu.de}%

\begin{abstract}
We consider the pricing and hedging of counterparty credit risk and funding when there is no possibility to hedge the jump to default of either the bank or the counterparty. This represents the situation which is most often encountered in practice, due to the absence of quoted corporate bonds or CDS contracts written on the counterparty and the difficulty for the bank to buy/sell protection on her own default. We apply local risk-minimization to find the optimal strategy and compute it via a BSDE. 
\end{abstract}


\date{\today}

\maketitle

\section{Introduction}\label{intro} 
In this paper we study pricing and hedging of counterparty credit risk and funding in absence of hedging possibilities against the default of the bank or the counterparty. We tackle the market incompleteness due to the presence of possible defaults with the well-known local risk-minimization approach extended to a multi-curve setting. We describe the optimal strategy via the solution of a BSDE and use this result to derive a decomposition of the price in terms of value adjustments.

The 2007-2009 financial crisis provided the motivation for a critical review of several assumptions regarding the valuation of financial products. According to the market practice and the academic literature the value of a product should account for the possibility of default of any agent involved in the transaction. The presence of multiple interest rate curves for funding also represents a significant shift with respect to a classical single-curve asset pricing theory. Such aspects are captured by value adjustments (xVA), which are added or subtracted to an idealized reference price, in order to account for the aforementioned frictions. The first contribution on counterparty risk is \cite{duhu96}. Before the 2007--2009 financial crisis, \cite{brigoMasetti} and \cite{cherubini05} studied the \emph{credit value adjustment} (CVA). Subsequently, a \emph{debt value adjustment} (DVA) was introduced in order to account for the default of the trader in, among others, \cite{bripapa11} and \cite{bricapa14}.

By funding costs/multiple curves we mean that nowadays it is recognized, in line with realistic market conditions, that the trader can fund her trading activity by means of a multitude of sources of funding involving different interest rates. In a Black-Scholes economy without default risk \cite{pit10} provides valuation formulas in presence or absence of collateral agreements, introducing the so-called \emph{funding value adjustment} (FVA). Such results are generalized in \cite{BieRut15} and, in a cross currency setting in \cite{fushita10b},  \cite{fushita09}, and \cite{gnoatto2020}. Further discussions on the FVA are 
\cite{papebri2011} and \cite{bripa2014ccp} where the role of central counterparties is discussed as well.

Funding and counterparty risk have been simultaneously considered afterwards, giving rise to the so-called xVAs, where the x represents a placeholder for different valuation adjustments. In a series of papers, Burgard and Kjaer generalize the classical Black-Scholes replication approach to include many effects, 
see \cite{bj2011a} and \cite{bj2013}. However, their results are based on replication arguments. The discounting approach of Brigo and co-authors does not assume market completeness, see \cite{BRIGO2019788}, \cite{bbfpr2022} and, for a mathematical formalization of the involved pricing equations \cite{bgk2024}. The inclusion of xVAs results in general in non-linear pricing equations satisfying \emph{backward stochastic differential equations} (BSDEs). Such an approach is studied in \cite{crepey2015a}, \cite{crepey2015b}, 
\cite{BiCaStu2018}, \cite{BiCaStu2019} and \cite{Biagini_Gnoatto_Oliva_2021}.

In practice it is usually not possible to hedge against the default of the bank or the counterparty. However, the vast majority of the literature assumes a complete market setting. This is particulary evident in the initial contributions of Burgard and Kjaer, focusing on extensions of the classical Black-Scholes-type replication argument. In concrete terms, computing e.g. CVA under market completeness is based on the assumption that there exists some traded asset such as a corporate bond or a CDS, such that the trading desk can hedge the jump to default of the counterparty. This assumption is justified only for a very limited set of bigger corporates. For most counterparties instead there are no hedging instruments against the jump to default. This has led to the introduction of different methodologies for the construction of artificial term structure of risk neutral default probabilities, for example \cite{nomura2013} or even the application of historical default probabilities. However such approaches do not offer any protection against default events.

Regarding the hedging of the bank's default this is also problematic from a practical point of view: in concrete cases banks have been proxy-hedging their DVA by trading the CDS of peers i.e. of other investment banks. Such a practice can lead to an approximate hedge before default but does not offer protection against the jump to default. For a bank, buying protection on her own default could even be prohibited under certain jurisdictions.

Looking again at the existing literature, the aforementioned discounting approach of Brigo does not require market completeness: only the martingale property of certain cumulative discounted cashflows is postulated. A similar route that does not postulate replication has been followed by Cr\'epey and co-authors from 2017 onward. In a series of papers, an analysis of xVA from the point of view of a stylized model of the balance sheet of the bank has been proposed in \cite{alcacre2017} \cite{Albanese02012021} \cite{css2020} \cite{crep2022}. An interesting approach based on local risk-minimization, limited to CVA, is performed in \cite{Bo2019}.

To summarize, default risk is in general difficult to hedge in concrete settings, up to the point where counterparty credit risk is essentially similar to an actuarial risk that is warehoused by the bank, in a fashion similar to an insurance company: the CVA is charged as an insurance policy premium to the counterparty and the default events erode a risk margin.

In this paper we study xVA valuation under market incompleteness by means of the local risk-minimization approach. The first study applying local risk-minimization to counterparty risk is \cite{Bo2019}, however the study is limited to CVA in a single curve framework. In our approach, we obtain a pricing-hedging BSDE in a full multiple-curve setting with realistic collateral specifications and bank defaultability, thus providing a full xVA framework. The system of possibly high-dimensional BSDEs we obtain is amenable to numerical treatment by means of deep learning-based approaches, see \cite{Han2018}, \cite{weinar2017}, \cite{Gnoatto2023} and \cite{Gnoatto2024}. While in the present paper we consider the pricing of a single contingent claim, our approach can be readily extended to account for the presence of a portfolio of derivatives between the bank and the counterparty, as done in \cite{Biagini_Gnoatto_Oliva_2021}, where xVA under arbitrary aggregation levels (margin sets and netting sets) is analyzed. Our xVA metrics consider only two agents, the bank and the counterparty, in line with, among others, \cite{BiCaStu2018} or \cite{bgk2024} and the current market practice. We deliberately adopt the perspective of a ring-fenced trading desk pricing a transaction at the netting-set level with a specific counterparty. We acknowledge the ongoing debate in the literature regarding whether valuation adjustments should be evaluated at the level of the bank’s entire portfolio to account for embedded nonlinearities or wealth transfers between shareholders and bondholders as in \cite{ads2019}. A fully holistic approach requires modeling the bank's central treasury explicitly, accounting for the Treasury CVA and Treasury DVA incurred when borrowing externally to fund the trading desk. Depending on the treasury's net funding position, these treasury-level adjustments can offset the trading desk's DVA or reclassify it as a pure funding benefit. However, incorporating these macro-level treasury dynamics and global debt networks shifts the focus from pricing a bilateral derivative to solving a corporate finance and capital allocation problem. Because our primary objective is to isolate and manage the unhedgeable risk of counterparty and bank default jumps, we restrict our scope to the bilateral netting set. This reflects the market reality that a counterparty's xVA charge cannot depend on the bank's opaque, firm-wide netting sets with third parties.

The paper is outlined as follows: in Section \ref{Sect:preliminaries} we introduce the financial setting. Here we make the explicit choice of presenting a canonical complete financial market model in the multiple curve setting in order to isolate the fact that incompleteness is introduced only by means of the default events. Section \ref{Sec:LocalRiskMultiCurve} presents local risk-minimization in the multiple curve setting as a tool to solve the pricing problem, whereas Section \ref{Sec:FSdecBSDE} connects local risk-minimization to BSDEs and derives our all-inclusive price decomposition in terms of an idealized clean price and valuation adjustments. Section \ref{Sec:KVA} finally provides a tentative definition of KVA arising from our approach.

\section{The financial setting}\label{Sect:preliminaries}
We fix a time horizon $T<\infty$ for the trading activity. We consider two agents named the \emph{bank} (B) and the 
\emph{counterparty} (C). Unless otherwise stated, throughout the paper we assume the bank's perspective 
and refer to the bank as the \emph{hedger.} 
All processes are modeled over a probability space $(\Omega, \mathcal{G}, \mathbb{P})$. \\
Let $W^f=(W_t^f)_{t \in [0,T]}$ be a $d$-dimensional Brownian motion on $(\Omega, \mathcal{G}, \mathbb{P})$ and $\mathbb{F}=(\mathcal{F}_t)_{t \in [0,T]}$ its natural filtration.
We denote by $\tau^B,$ respectively by $\tau^C,$ the \emph{time of default} of the bank, respectively of the counterparty,
and by $H^j=(H^j_t)_{t \in [0,T]} $ with $H^j_t:=\textbf{1}_{\lbrace \tau^j\leq t \rbrace}, \, j\,\in\, \{B,C\},$ the associated jump process.  Let $\mathbb{H}^j = (\mathcal{H}_t^j)_{t \,\in\, [0,T]}, \; j \,\in\, \{B,C\},$ be the natural filtration of $H^B,\, H^C,$ respectively. On $(\Omega, \cG,\PP)$ we consider the filtration $\GG=(\mathcal{G}_t)_{t \in [0,T]}$ given by
$$
\mathcal{G}_t = \mathcal{F}_t \vee \mathcal{H}^B_t \vee \mathcal{H}^C_t, \quad t \in [0,T],
$$ and set $\mathcal{G}:=\mathcal{G}_T.$ All filtrations are required to satisfy the usual hypotheses of completeness and right-continuity. 
Let $$
F_t^j:= \mathbb{P}(\tau^j \leq t \vert \mathcal{F}_t), \quad t \in [0,T], \; j \,\in\, \{B,C\},
$$
be the conditional distribution function of the default time $\tau^j$ with  $F_t^j<1$ for $j\in \{B,C \}$ and $t\in [0,T].$
The associated hazard process $\Gamma^j=(\Gamma^j_t)_{t \in [0,T]}$ of $\tau^j$ under $\mathbb{P}$ is defined by 
$$\Gamma^{j}_t:=-\ln\left(1-F_t^j\right),\quad t \in [0,T], \; j \in \lbrace B, C \rbrace.$$
In particular, we assume that the hazard process $\Gamma^j$ admits the following presentation
$$\Gamma^j_t=\int_0^t \lambda^{j}_s ds, \quad t \,\in\, [0,T],\, j\,\in\,\{B,C\},$$   
where $\lambda^j=(\lambda^j_t)_{t \in [0,T]}$ is an $\FF$-adapted, non-negative bounded process. Note that this guarantees that the $\mathbb{F}$-hazard process $\Gamma^j$ is continuous and increasing and $\tau^j$ is a $\mathbb{G}$-stopping time.
Furthermore, $\tau^B$ and $\tau^C$ are assumed to satisfy $\mathbb{P}(\tau^B=\tau^C)=0$ and
we introduce the first default time $\tau$ by
\begin{equation}\label{eq:tau}
    \tau:= \tau^B \wedge \tau^C.
\end{equation}
%We assume that $\tau$ has an increasing $\mathbb{F}$-hazard process $\Gamma=(\Gamma_t)_{t \in [0,T]}.$ This is for example satisfied if $\tau^B$ and $\tau^C$ are conditionally independent with respect to $\mathbb{F}$ under $\mathbb{P}$, see Lemma 9.1.3 in \cite{BieRut15}. \\

In the present paper we will make use of the so-called \emph{Immersion Hypothesis.}

\begin{assumptionStanding} \label{hp:H}   
Any $(\FF,\PP)$-martingale is a $(\GG,\PP)$-martingale. 
\end{assumptionStanding}

\begin{remark}  \label{remark:PropertiesEnlargedFiltration}
For $j \in \lbrace B, C \rbrace$ the process $M^j=(M_t^j)_{t \in [0,T]}$ defined by 
    \begin{equation} \label{eq:Compensated}
        M_t^j:=H_t^j-\Gamma^j_{t \wedge \tau^j}= H_t^j-\int_0^t \tilde{\lambda}_u^j du, \quad t \in [0,T],
    \end{equation}
    with $\tilde{\lambda}_u^j:=\lambda_u^j \textbf{1}_{\lbrace{ \tau^j \geq u \rbrace}}$ for $u\in [0,T]$, is a $(\mathbb{G},\mathbb{P})$-martingale, see Lemma 9.1.1 in \cite{BieRut2004}. Furthermore, in this framework the predictable representation property in $\mathbb{G}$ holds with respect to $W^f, M^B, M^C$, see Corollary 5.2.4 in \cite{BieRut15}. Moreover, by Hypothesis \ref{hp:H} $W^f$ is a $\mathbb{G}$-martingale and by the optional stopping and L\'evy's theorem we conclude that $W^{\tau,f}$ is a Brownian motion on $[0, \tau \wedge T]$.
\end{remark}
Before presenting the market model in detail, we introduce some notation used throughout the paper.
\begin{notation} \label{notation:SpacesLocalRisk0}
\begin{enumerate}
\item We denote by $L^2(\mathcal{F}_T, \mathbb{P})$ the space of all $\mathcal{F}_T$-measurable random variables $Y$ such that $\mathbb{E}_{\mathbb{P}}\left[\vert Y \vert^2\right]<\infty.$
\item We denote by $L^2(\mathbb{F},\mathbb{P})$ the space of all right-continuous, $\mathbb{F}$-adapted, $\mathbb{R}$-valued  processes $Y=(Y_t)_{t \in [0,T]},$ such that $Y_t \in L^2(\mathcal{F}_t, \mathbb{P})$ for all $t \in [0,T].$
\item We denote by $\mathcal{S}^2(\mathbb{F},\mathbb{P})$ the space of all $\mathbb{R}$-valued square-integrable $(\mathbb{F}, \mathbb{P})$-semimartingales $Y=(Y_t)_{t\in [0,T]}$ that can be decomposed as follows
\begin{align} \label{eq:SemimartingaleDecomposition}    Y_t=Y_0+M_t^Y+A_t^Y, \quad t \in [0,T],
\end{align}
where $Y_0 \in \mathbb{R}$, $M^Y=(M^Y_t)_{t \in [0,T]}$ is an $\mathbb{R}$-valued, square-integrable $(\mathbb{F}, \mathbb{P})$-martingale and $A^Y=(A^Y_t)_{t \in [0,T]}$ is an $\mathbb{R}$-valued, $\mathbb{F}$-predictable process of finite variation such that $M^Y_0=A^Y_0=0.$ The corresponding space of $\mathbb{R}^d$-valued semimartingales is denoted by $\mathcal{S}^2_{\textnormal{d}}(\mathbb{F},\mathbb{P})$.
\end{enumerate} 
The spaces $L^2(\mathcal{G}_T,\mathbb{P}), L^2(\mathbb{G},\mathbb{P})$ and $\mathcal{S}^2_{\textnormal{d}}(\mathbb{G}, \mathbb{P})$ are defined analogously, but with respect to $\mathcal{G}_T$ and $\mathbb{G}.$ 
\end{notation}

\subsection{Traded assets}  \label{sec:BasicTradedAssets}
For $d \geq 1,$ we denote by $S^i=(S^i_t)_{t \in [0,T]}$, $i=1,\ldots, d,$ the \emph{ex-dividend price} (i.e. the price) of the risky asset $i$, which is assumed to %be a {continuous $(\mathbb{F}, \mathbb{P})$-semimartingale.} More specifically, the dynamics of $S^i=(S_t^i)_{t \in [0,T]}$ are given by
satisfy the following dynamics
\begin{align} \label{eq:dynamicSspecifiedDefaultsTheorem}
    \begin{cases}
dS_t^i&=S_t^i\left(\mu^i(t,S_t)dt+{\sigma}^i(t,S_t)dW^{f,i}_t \right)\\
S_0^i&=s_0^i >0
\end{cases}
\end{align}
on $[0,T]$ with 
\begin{align*}
   \mu^i&:\left([0,T] \times \mathbb{R}^d, \mathcal{B}\left([0,T] \times \mathbb{R}^d\right)\right) \to \left(\mathbb{R}, \mathcal{B}(\mathbb{R})\right),\\ 
   \sigma^i&:\left([0,T] \times \mathbb{R}^d, \mathcal{B}\left([0,T] \times \mathbb{R}^d\right)\right) \to \left(\mathbb{R}_+, \mathcal{B}(\mathbb{R}_+)\right)
\end{align*}
with $\sigma^i\neq 0$ for $i=1,...,d,$ and such that there exists a unique strong solution of \eqref{eq:dynamicSspecifiedDefaultsTheorem}. 
As we consider an incomplete market model we do not assume that there exist risky bonds issued by the bank or the counterparty, which would allow to hedge the jump to default. \\
 In the following, we assume the existence of an indexed finite family of cash accounts $(B^x)_{x \,\in\,\tI},$ where the stochastic process 
$r^x= (r^x_t)_{t \in [0,T]}$ is bounded and $\FF$-adapted for all $x \,\in\, \tI.$ In particular, we assume that there exists a constant ${K}_r>0$ such that for all $x \in \mathcal{I}$
\begin{align}  \label{eq:BoundsInterestRates}
\left \vert r^x \right \vert \leq K_r \quad \mathbb{P}\text{-a.s.}
\end{align}
Note that we can always define such a constant $K_r$ as the supremum over the absolute value of all upper and lower bounds for $r^x$ with $x \in \mathcal{I}$. Moreover, in the following we denote by $C(K_r)>0$ a generic constant depending on $K_r.$
The set of indices $\tI$ embodies the type of agreement the counterparties establish in order to mitigate the counterparty credit risk. We will specify the characteristics of the aforementioned indices later on.  

All cash accounts, with unitary value at time $0$, are assumed to be strictly positive 
continuous processes of finite variation of the form 
\begin{align} \label{def:cash_account}
B^x_t := e^{\int_0^t r^{x}_s ds}, \quad t \,\in\, [0,T], \, x \in \mathcal{I}. 
\end{align}

\subsection{Repo-trading} \label{repo_trading}
In line with the existing literature, we assume that the trading activity on the risky assets $S$, \textcolor{black}{introduced in Subsection \ref{sec:BasicTradedAssets},} is \emph{collateralized.} 
This means that borrowing and lending activities related to risky securities are financed via security lending or \emph{repo market,} 
\textcolor{black}{see} \cite{BiCaStu2018}. \textcolor{black}{To model this market we introduce for every risky asset $S^i,$ $i=1,...,d,$ a repo account $B^i=(B^i_t)_{t \in [0,T]}$ such that 
\begin{equation} \label{eq:DynamicsRepoAccount}
dB_t^i = r_t^i B_t^i dt, \quad  t \in [0,T],
\end{equation}
where $r^i=(r_t^i)_{t \in [0,T]}$ denotes the repo rate for account $i$.} 
Since transactions on the repo market are collateralized by the risky assets, repo rates are lower than unsecured funding rates. 
As argued in \cite{crepey2015a}, assuming that all assets are traded via repo markets is not restrictive. In the setting of this paper we always assume that the transactions are fully collateralized, i.e. that
\begin{align} \label{eq:repoConstraint}
  \xi_t^i S_t^i+ \psi_t^i B_t^i=0, \quad  t \in  [0,T], \, i=1,...,d,
\end{align}
where $\xi^i$ and $\psi^i$ denote the number of shares invested in the risky asset $S^i$ and the repo account $B^i$, respectively. We refer to \eqref{eq:repoConstraint} as the \emph{repo-trading constraint}.

It is worth noting that $\xi_t^i, \, i = 1, \ldots, d,$ may be either positive or negative. Here $\xi_t^i >0$ means that we are in a long position, which has to be financed by collateralization. On the other hand, $\xi_t^i < 0$ implies that the $i$-th asset is shorted, so that the whole amount of collateral is deposited  in the riskless asset. 

\subsection{Market model}
From now on, we assume for the market model $(S^1,...,S^d,B^1,...,B^d)$ that for every $i=1,...,d$ there exists a uniform constant $K_i>0$ such that for all $t \in [0,T]$
\begin{align} \label{eq:BoundsExistenceSolutionDefaults}
    \left \vert \frac{\mu^i(t,S_t)-r_t^i}{\sigma^{i}(t,S_t)} \right \vert \leq K_i \quad \mathbb{P}\text{-a.s.}
\end{align}
and
\begin{align} \label{eq:SetNotEmpty}
    \mathbb{E}_{\mathbb{P}} \left[ \int_0^T \left( \vert \mu^i(u,S_u) \vert + \left(\sigma^i (u,S_u)\right)^2\right) du\right]<\infty.
\end{align}
Next, we define the $\mathbb{R}^d$-valued discounted asset price process $\widetilde{S}=(\widetilde{S}_t)_{t \in [0,T]}$ by
    \begin{align} \label{eq:DiscountedPrice}
        \widetilde{S}_t^{i}:= \frac{S_t^i}{B_t^i}, \quad t \in [0,T], \; i=1,...,d,
    \end{align}
which satisfies 
\begin{align}  \label{eq:DyS}
d \widetilde{S}_{t}^{i}&=\frac{1}{B_t^i}\left( dS_t^i -r_t^i S_t^i dt \right) = \widetilde{S}_{t}^{i}\left(\left(\mu^i(t,S_t)-r_t^i\right)dt+{\sigma}^i(t,S_t)dW^{f,i}_t \right).  
\end{align}
For all $i=1,...,d$ we choose $\mu^i, \sigma^i,r^i$ such that $ \widetilde{S} \in \mathcal{S}^2_{\textnormal{d}}(\mathbb{F},\mathbb{P})$. 
This means that for all $i=1,...,d,$ $\widetilde{S}^{i}$ is an $\mathbb{R}$-valued, square-integrable $(\mathbb{F}, \mathbb{P})$-semimartingale with canonical decomposition
\begin{align}  \label{eq:SemimartingaleDecompositionS}
\widetilde{S}_t^i&=\widetilde{S}_0^i+M_t^{\widetilde{S}^i}+A_t^{\widetilde{S}^{i}} = \widetilde{S}_0^i + \int_0^t  \widetilde{S}_{u}^{i} \sigma^i(u,S_u)dW_u^{f,i}  +\int_0^t \widetilde{S}_{u}^{i}\left(\mu^i(u,S_u)-r_u^i\right)du , \quad 
t \in [0,T],
\end{align} 
where $M^{\widetilde{S}^i}=(M^{\widetilde{S}^i}_t)_{t \in [0,T]}$ is an $\mathbb{R}$-valued, square-integrable $(\mathbb{F}, \mathbb{P})$-martingale and $A^{\widetilde{S}^i}=(A^{\widetilde{S}^i}_t)_{t \in [0,T]}$ is an $\mathbb{R}$-valued, $\mathbb{F}$-predictable finite variation process. By Remark \ref{remark:PropertiesEnlargedFiltration} and as $\widetilde{S} \in \mathcal{S}_{\textnormal{d}}^2(\mathbb{F}, \mathbb{P})$ we conclude that the stopped process
 $\widetilde{S}^\tau=(\widetilde{S}_{\tau \wedge t})_{t \in [0,\tau \wedge T]}$ is an element in $\mathcal{S}^2_{\textnormal{d}}(\mathbb{G},\mathbb{P})$. 
\begin{remark}
Note that it is also possible to weaken the integrability conditions on $\widetilde{S}$ by assuming that $\widetilde{S} $ is in the space of all locally square-integrable $(\mathbb{F}, \mathbb{P})$-semimartingales, denoted by $\mathcal{S}^2_{\textnormal{loc,d}}(\mathbb{F},\mathbb{P})$. 
  However, to avoid technicalities we work with the space $\mathcal{S}^2_{\textnormal{d}}(\mathbb{F},\mathbb{P})$.
\end{remark}
We denote by $\mathcal{P}^2_{\mathbb{F}}(\widetilde{S})$ the set of all equivalent probability measures $\overline{\mathbb{Q}}$ on $(\Omega, \mathcal{F}_T)$ such that for all $i=1,...,d$ the process $\widetilde{S}^i$ defined in \eqref{eq:DiscountedPrice} is an $(\mathbb{F}, \overline{\mathbb{Q}})$-martingale and $\frac{d \overline{\mathbb{Q}}}{d\mathbb{P}}\big \vert_{\mathcal{F}_T} \in L^2(\Omega,\mathcal{F}_T)$. The conditions in \eqref{eq:BoundsExistenceSolutionDefaults} and \eqref{eq:SetNotEmpty} guarantee that $\mathcal{P}_{\mathbb{F}}^2(\widetilde{S}) \neq \emptyset.$ \\ On $(\Omega, \mathcal{G}_T)$ we define the probability measure ${\mathbb{Q}}$ by the density process 
\begin{align}
    {\mathscr{Z}}_t:=\frac{d{\mathbb{Q}}}{d\mathbb{P}}\Bigg \vert_{\mathcal{G}_t}&=  \mathscr{E}\left( \sum_{i=1}^d \int_0^{\cdot} -\frac{\mu^i(u,S_u)-r_u^i}{\sigma^{i}(u,S_u)} dW_u^{f,i}\right)_t \nonumber \\
    &=\exp\left(\sum_{i=1}^d \int_0^{t} \frac{\mu^i(u,S_u)-r_u^i}{\sigma^{i}(u,S_u)} dW_u^{f,i}- \frac{1}{2} \sum_{i=1}^d \int_0^{t} \left(\frac{\mu^i(u,S_u)-r_u^i}{\sigma^{i}(u,S_u)} \right)^2 du\right), \quad t \in [0,T].
    \label{eq:StochasticExponential}
\end{align}
 By \eqref{eq:BoundsExistenceSolutionDefaults} we obtain that Novikov's condition holds, i.e.  
\begin{align*}
\mathbb{E}_{\mathbb{P}}\left[  \exp \left(\frac{1}{2}\int_0^T \sum_{i=1}^d \left( \frac{\mu^i(u,S_u)-r_u^i}{\sigma^{i}(u,S_u)} \right)^2 du \right) \right]\leq e^{Td(\max_{i=1,...,d} K_i^2)}<\infty,
\end{align*}
and thus $ {\mathscr{Z}}$ is a positive $(\mathbb{G}, \mathbb{P})$-martingale. As $\mathscr{Z}$ is $\mathbb{F}$-adapted, $\mathscr{Z}$ is also a $(\mathbb{F},\mathbb{P})$-martingale.
Moreover, $ {\mathscr{Z}} \in L^2(\mathbb{F}, \mathbb{P}) \subseteq L^2(\mathbb{G}, \mathbb{P})$, as by \eqref{eq:BoundsExistenceSolutionDefaults} we have that 
\begin{align} \label{eq:SecondMoment}
\mathbb{E}_{\mathbb{P}}\left[ \vert  {\mathscr{Z}}_t \vert^2\right] \leq d \max_{i=1,...,d} K_i, \quad t \in [0,T]
\end{align}
see e.g. Proposition 2.5.1 in \cite{delong}. Moreover, $\mathbb{Q}$ is the minimal martingale measure\footnote{A martingale measure for $\widetilde{S}$, denoted by $\mathbb{Q}$, which is equivalent to $\mathbb{P}$ such that $\frac{d\mathbb{Q}}{d\mathbb{P}} \in L^2(\mathcal{G}_T, \mathbb{P})$ is called  \emph{minimal} for $\widetilde{S}$ if any square-integrable $(\mathbb{G},\mathbb{P})$-martingale, which is strongly orthogonal to the martingale part of $\widetilde{S}$ under $\mathbb{P}$ is also a $(\mathbb{G},\mathbb{Q})$-martingale.} for $S$, see e.g. \cite{schweizer_survey}.
\begin{remark}
    Condition \eqref{eq:BoundsExistenceSolutionDefaults} guarantees the existence of the minimal martingale measure for $S.$ Furthermore, it also establishes the existence and uniqueness of a BSDE later on, see Theorem \ref{theorem:LinkBSDE}.  
\end{remark}

\subsection{Funding}
Within the bank, the trading desk borrows and lends money from/to the treasury desk. For simplicity, we assume that the borrowing and lending rates do not differ, i.e. the rate at which 
the trading desk borrows from and lends to the treasury desk are the same and this rate is denoted by $r^f=(r_t^f)_{t \in [0,T]}$ and the associated treasury account by $B^f=(B_t^f)_{t \in [0,T]}$ with 
$$dB_t^f =r_t^fB_t^f dt, \quad t \in [0,T].$$
As we will discount also with the treasury account $B^f$ we define the following $\mathbb{R}^d$-valued process $X=(X_t)_{t \in [0,T]}$ by 
\begin{equation} \label{eq:StochasticIntegral}
X_t^i :=\int_{0}^t \frac{B^i_u}{B_u^f} d \widetilde{S}_u^i, \quad  t \in [0,T],  \; i=1,...,d,
\end{equation}
which is continuous and in the space $\mathcal{S}^2_{\textnormal{d}}(\mathbb{F},\mathbb{P}).$ This follows as $\widetilde{S} \in \mathcal{S}_{\textnormal{d}}^2(\mathbb{F}, \mathbb{P})$ and thus
for every $i=1,...,d$ we have 
\begin{align} 
&\mathbb{E}_{\mathbb{P}}\left[ \int_0^T \left(\frac{B^i_u}{B^f_u} \right)^2d \langle M^{\widetilde{S}^i} \rangle_u + \left(\int_0^T \left \vert \frac{B^i_u}{B^f_u}  dA^{\widetilde{S}^i}_u \right \vert \right)^2   \right] \nonumber \\
&=\mathbb{E}_{\mathbb{P}}\left[ \int_0^T \left(\frac{B^i_u}{B^f_u} \right)^2  \left(\widetilde{S}_{u}^{i} \sigma^i(u,S_u)\right)^2 du  + \left(\int_0^T \left \vert \frac{B^i_u}{B^f_u}  \widetilde{S}_{u}^{i}\left(\mu^i(u,S_u)-r_u^i\right)du \right \vert \right)^2   \right]\nonumber \\
&\leq C(K_r) \mathbb{E}_{\mathbb{P}}\left[\int_0^T  \left(\widetilde{S}_{u}^{i} \sigma^i(u,S_u)\right)^2 du  + \left(\int_0^T \left \vert   \widetilde{S}_{u}^{i}\left(\mu^i(u,S_u)-r_u^i\right)du \right \vert \right)^2   \right]< \infty, \nonumber 
\end{align}
where we use \eqref{eq:BoundsInterestRates}.
%This follows by the Lemma on page 175 in \cite{protter}.

For $i=1,...,d$ the dynamics of the stopped process $X^{\tau,i}$ are given by 
 \begin{align}
        X_t^{\tau,i}& =X_0^i +A_t^{X^{\tau,i}}+ M_t^{X^{\tau,i}}\nonumber \\
        &=  X_0^i + \int_0^{t}  \frac{S_u^{\tau,i}}{B_u^f}  \left(\mu^i(u,S_u^{\tau}) -r_u^i \right) du +\int_0^{  t}  \frac{S_u^{\tau,i}}{B_u^f} {\sigma}^{i}(u,S_u^{\tau}) dW^{\tau,f,i}_u, \quad t \in [0,\tau \wedge T].\label{eq:SemimartingaleDecomposotionX}
    \end{align} 
As $X$ is an element in $\mathcal{S}^2_{\textnormal{d}}(\mathbb{F},\mathbb{P})$ it follows by Remark \ref{remark:PropertiesEnlargedFiltration} and as $X \in \mathcal{S}_{\textnormal{d}}^2(\mathbb{F}, \mathbb{P})$ 
 that the stopped process $X^{\tau}=(X_{\tau \wedge t})_{t \in [0, \tau \wedge T]}$ is an element in $\mathcal{S}^2_{\textnormal{d}}(\mathbb{G},\mathbb{P})$. 
\subsection{Financial contract and collateralization}
In the following, we assume that the bank enters into an over-the-counter derivative contract at time $0$ with a counterparty in the market. 
\begin{definition}
 The payment stream of a \emph{non-defaultable and uncollateralized financial contract} is represented by an ${\FF}$-adapted c\`adl\`ag process of finite variation $A = (A_t)_{t \,\in\, [0,T]}. $
\end{definition}
We use the notation $\Delta A_t:=A_t-A_{t-}$ for the jumps of $A$. 
Note that $A_{\tau-}$ represents the last payment before default, see also \cite{BriMor2018}.\\
The cash flow of $A$ does not take into account the risk factors resulting from collateralization and the possible default of the bank or the counterparty.   \\
For what concerns financial contracts, collateralization is a method to minimize losses due to default 
of the counterparty by using margins. In the financial jargon, a \emph{margin} represents an economic value, either in the form of cash or risky securities, 
exchanged between the counterparties of a financial transaction, in order to reduce their risk exposures. \textcolor{black}{To keep the notation simple we do not distinguish between initial margin and collateral (or variation margin). This means that the collateral process $C=(C_t)_{t \in [0,T]}$ represents the overall collateralization. For a detailed study of the two different margins we refer to \cite{Biagini_Gnoatto_Oliva_2021}.} 

A collateral is posted between the bank and the counterparty to mitigate counterparty risk. 
The collateral process $C$ is assumed to be $\textcolor{black}{\FF}$-adapted. 
We follow the convention \textcolor{black}{of \cite{bbfpr2022}: }

\begin{itemize}
\item If \textcolor{black}{$C_t<0,$} we say that the bank is the \emph{collateral provider.} It means that the counterparty measures a positive exposure 
			towards the bank, so it is a potential lender to the bank, hence the bank provides/lends collateral to reduce its exposure. \textcolor{black}{The bank is remunerated at interest rate $r^{c,l}=(r^{c,l}_t)_{t \in [0,T]}$.}

\item If \textcolor{black}{$C_t>0,$} we say that the bank is the \emph{collateral taker.} It means that the bank measures a positive exposure towards 
			the counterparty, so it is a potential lender to the counterparty, hence the counterparty provides/lends collateral to reduce 
			its exposure. \textcolor{black}{In this case, the bank has to pay instantaneous interest rate $r^{c,b}=(r^{c,b}_t)_{t \in [0,T]}$.}
\end{itemize} 
In this case, the effective collateral accrual rate $\bar{r}^c=(\bar{r}_t^c)_{t \in [0,T]}$ equals
$$
\bar{r}_t^c:=r_t^{c,l} \textbf{1}_{\left\{C_t<0\right\}}+r_t^{c,b} \textbf{1}_{\left\{C_t \geq 0\right\}},  \quad  t\in [0,T].
$$
If there is a collateral agreement (or a multitude of agreements)  between the bank and the counterparty, in evaluating 
portfolio dynamics we need to make a distinction between the value of the portfolio and the wealth of the bank, the two 
concepts being distinguished since the bank is not the legal owner of the collateral (prior to default).

In this paper collateral is always posted in the form of cash, in line with standard \textcolor{black}{current market practice}. Moreover, we assume 
\emph{rehypothecation,} meaning that the holder of collateral can use the cash to finance her trading activity. 
This is the opposite of \emph{segregation,} where the received cash collateral must be kept in a separate account and 
cannot be used to finance the purchase of assets.

We simply set $\bar{r}^c=\rck{l}=\rck{b}$ in case there is no bid-offer spread in the collateral rate. 
Possible choices for the collateral rate are e.g. {\euro}STR for EUR trades, SOFR for USD and SONIA for GBP trades. 
Such rates are overnight rates with a negligible embedded risk component. The choice of such approximately risk-free rates
as collateral rates is motivated by market consensus. However, two counterparties might enter a collateral agreement 
that involves a remuneration of collateral at any other risky rate of their choice. 
Here we do not assume any requirements on collateral rates. This allows us to cover the quite common situation where the 
collateral rate agreed between the two counterparties in the CSA is defined by including a real valued 
spread over some market publicly observed rate, e.g. {\euro}STR $- 50\, bps,$ where $bps$ stands for \emph{basis points.}

In case of default, cashflows are exchanged between the surviving agent and the liquidators of the defaulted agent.
Here we use the term \emph{agent} as a placeholder for the bank or for the counterparty. Due to the exchange of 
cashflows at default time, agents need to perform a valuation of the position at a random time 
\textcolor{black}{represented by the close-out condition, see \cite{BiCaStu2018} Section 3.4}. 
The object of the analysis can be the value in the absence of counterparty risk (referred to in the literature as 
\emph{risk-free close-out}) or the value of the trade including the price adjustments due to counterparty risk and 
funding (\emph{risky close-out}), see e.g. \cite{BriMor2018}. 
A risky close-out condition guarantees that the surviving counterparty can ideally 
fully substitute the transaction with a new trade entered with another counterparty with the same credit quality. 
This comes at the price of a significant increase of the complexity of the valuation equations. Current 
market practice and the existing literature mainly focus on the case of risk-free close-out value. \\

We now introduce an auxiliary artificial interest rate process $r=(r_t)_{t \in[0,T]}$ which is assumed to be right-continuous, $\mathbb{F}$-adapted and bounded, i.e. \eqref{eq:BoundsInterestRates} holds. The associated bank account is denoted by $B^r=(B_t^r)_{t \in [0, T]}$. 
Following the approach of \cite{Biagini_Gnoatto_Oliva_2021}, we now introduce the definition of the clean value and choose  the minimal measure $\mathbb{Q}$ as a reference probability in this setting, as it  refers to the  assets only and leaves the dynamics of the process associated to $\tau$ unaffected.

\begin{definition} \label{def:CleanValue}
The \emph{clean value process} $\mathcal{V}=(\mathcal{V}_t)_{t \in [0,T]}$ of a financial contract $A$ is defined by
\begin{align*}
 \mathcal{V}_t:= \mathbb{E}_{\mathbb{Q}}\left[-\int_{]t,T]}  \frac{B_t^r}{B_u^r}\ d A_u \Big \vert \mathcal{F}_t \right]= \mathbb{E}_{\mathbb{P}}\left[-\int_{]t,T]}  {\mathscr{Z}}_u \frac{B_t^r}{B_u^r}\frac{1}{ {\mathscr{Z}}_t} d A_u \Big \vert \mathcal{F}_t \right], \quad t \in [0,T].
\end{align*}
The \emph{close-out valuation process} $Q=(Q_t)_{t \in [0, T]}$ of the financial contract $A$ is defined by
\begin{align*} 
 Q_t:=  \mathbb{E}_{\mathbb{Q}}\left[\int_{[t,T]}  \frac{B_t^r}{B_u^r}\ d A_u \Big \vert \mathcal{F}_t \right]=\mathbb{E}_{\mathbb{P}}\left[\int_{[t,T]}  {\mathscr{Z}}_u \frac{B_t^r}{B_u^r}\frac{1}{ {\mathscr{Z}}_t} d A_u \Big \vert \mathcal{F}_t \right]= -\mathcal{V}_t + \Delta A_t, \quad t \in [0,T].
\end{align*}
\end{definition}
\begin{remark} \label{remark:ConditioningGF}
\begin{enumerate}
\item Note that it holds for $t \in [0,  T]$
\begin{align} \label{eq:EquivaleneConditionalExpectations}
\mathcal{V}_t=\mathbb{E}_{\mathbb{P}}\left[-\int_{]t,T]}  {\mathscr{Z}}_u \frac{B_t^r}{B_u^r}\frac{1}{{\mathscr{Z}}_t} d A_u \Big \vert \mathcal{G}_t \right] \quad \text{and} \quad Q_t=  \mathbb{E}_{\mathbb{P}}\left[\int_{[t,T]}  {\mathscr{Z}}_u \frac{B_t^r}{B_u^r}\frac{1}{{\mathscr{Z}}_t} d A_u \Big \vert \mathcal{G}_t \right].
\end{align}
This follows directly by Hypothesis \ref{hp:H} and as $\int_{[t,T]} {\mathscr{Z}}_u \frac{B_t^r}{B_u^r}\frac{1}{{\mathscr{Z}}_t} d A_u $ is $\mathcal{F}_T$-measurable, see e.g. Lemma 3.2.1 in \cite{BieJeanRut2009}.
\item On the event $\lbrace{\tau \leq T \rbrace}$ we have 
\begin{align} \label{eq:CadlagNoJump}
Q_{\tau}=  \mathbb{E}_{\mathbb{P}}\left[\int_{[{\tau},T]} {\mathscr{Z}}_u \frac{B_{\tau}^r}{B_u^r}\frac{1}{{\mathscr{Z}}_{\tau}} d A_u \Big \vert \mathcal{G}_{\tau} \right]= -\mathcal{V}_{\tau} + \Delta A_{\tau}= -\mathcal{V}_{\tau},
\end{align}
where we use \eqref{eq:EquivaleneConditionalExpectations} and that $\Delta A_{\tau}=0$ by Lemma 2.1 in \cite{crepey2015b}. 
\end{enumerate}
\end{remark}

In the intensity-based setting of this paper,  $\Delta A_\tau=0$. This would be restrictive in situations where the link between credit exposures and default probabilities, known in practice as wrong-way risk, is significant. An example of such situation would be the pricing of xVA for credit derivatives. This problem has already been addressed in the literature via a variety of approaches: \cite{BrigoVrins2018} tackle wrong-way risk by means of measure  change techniques, while \cite{Jeanblanc2017,Vrins2026} consider models where the Az\'ema supermartingale features a martingale part. This last approach is explicitly linked in \cite{Vrins2026} to the dynamic copula models used in \cite{Crpey2016}.
However, the main purpose of our article is to study the impact of market incompleteness due to the fact that hedging instruments against the jump to default do not exist in practice in most cases, and for this purpose it is sufficient to use the standard reduced-form approach.

We set $C_t := g(\cV_t), \, t \,\in\, [0,T],$ where $g: \, \RR \,\rightarrow\, \RR$ is a Lipschitz function, which allows to cover realistic collateral specifications, see e.g. \cite{listag2015} and \cite{BaFuMa2019}. Moreover, note that 
\begin{align} \label{eq:NoCollateralPosted}
C_{\tau \wedge T}=C_T \textbf{1}_{\lbrace \tau >T \rbrace} + C_{\tau}\textbf{1}_{\lbrace \tau \leq T \rbrace}=C_T \textbf{1}_{\lbrace \tau >T \rbrace} + C_{\tau-}\textbf{1}_{\lbrace \tau \leq T \rbrace},
\end{align}
as the collateral is not updated at the time of default.

Next, we provide the definition of a defaultable, collateralized contract.
\begin{definition}
%\begin{enumerate}
     A \emph{defaultable, collateralized contract} $(A,R, C, \tau)$ is a contract with the promised cashflow process $A$, collateral account $C$, and a close-out payoff $\vartheta_{\tau}(R,C)$ occurring at the first default time $\tau$ given by 
    \begin{align} \label{def:CloseOut}
        \vartheta_\tau(R,C):=R_{\tau}+C_{\tau-},
    \end{align}
    where the recovery payoff $R_{\tau}$ is defined on the event $\lbrace \tau \leq T \rbrace$ by 
\begin{equation} \label{eq:finalCloseOut} 
\begin{aligned}
	R_\tau & := \textbf{1}_{\lbrace{\tau^C< \tau^B\rbrace}}(R^C \mathcal{Y}^+-\mathcal{Y}^-)+\textbf{1}_{\lbrace{\tau^B< \tau^C\rbrace}}(\mathcal{Y}^+-R^B \mathcal{Y}^-),
\end{aligned}
\end{equation} 
with $\mathcal{Y}:=Q_{\tau}-C_{\tau -}$ and the recovery rates $0<R^j<1, j \in \lbrace B, C \rbrace,$ of the bank and the counterparty, respectively.
Equation \eqref{def:CloseOut} is known as \emph{close-out condition}.
\end{definition}

Note that we assume the recovery rates to be constant. However this can be extended to the more general case of stochastic rates by imposing some additional measurability and integrability conditions so that the following results continue to hold.
\begin{remark}
 We can rewrite the close-out condition $\vartheta_{\tau}(R,C)$ in Definition \ref{def:CloseOut} with $\mathcal{Y}=Q_{\tau}-C_{\tau_-}$ as follows
\begin{align}
\vartheta_{\tau}(R,C)&= R_{\tau} + C_{\tau-} \nonumber \\
&=\textbf{1}_{\lbrace{\tau^C< \tau^B\rbrace}}(R^C \mathcal{Y}^+-\mathcal{Y}^-)+\textbf{1}_{\lbrace{\tau^B< \tau^C\rbrace}}(\mathcal{Y}^+-R^B \mathcal{Y}^-)+C_{\tau-} \nonumber \\
&=\textbf{1}_{\lbrace{\tau^C< \tau^B\rbrace}}(R^C \mathcal{Y}^++\mathcal{Y}-\mathcal{Y}^+)+\textbf{1}_{\lbrace{\tau^B< \tau^C\rbrace}} (\mathcal{Y}+\mathcal{Y}^--R^B \mathcal{Y}^-)+C_{\tau-}  \nonumber \\
&=- \textbf{1}_{\lbrace{\tau^C< \tau^B\rbrace}}(1-R^C) \mathcal{Y}^+ +\textbf{1}_{\lbrace{\tau^B< \tau^C\rbrace}} (1-R^B)\mathcal{Y}^- +\mathcal{Y}\left(\textbf{1}_{\lbrace{\tau^C< \tau^B\rbrace}}+\textbf{1}_{\lbrace{\tau^B< \tau^C\rbrace}}\right) +C_{\tau-} \nonumber \\
&=- \textbf{1}_{\lbrace{\tau^C< \tau^B\rbrace}}(1-R^C) (Q_{\tau}-C_{\tau-})^+ +\textbf{1}_{\lbrace{\tau^B< \tau^C\rbrace}} (1-R^B)(Q_{\tau}-C_{\tau-})^- +Q_{\tau}-C_{\tau-} +C_{\tau-} \label{eq:ExcludeSameStoppingTimes} \\
&=Q_{\tau}- \textbf{1}_{\lbrace{\tau^C< \tau^B\rbrace}}(1-R^C) (Q_{\tau}-C_{\tau-})^+ +\textbf{1}_{\lbrace{\tau^B< \tau^C\rbrace}} (1-R^B)(Q_{\tau}-C_{\tau-})^-, \label{eq:PayoffRewriten}
\end{align}
where we use in \eqref{eq:ExcludeSameStoppingTimes} that $\mathbb{P}(\tau^B=\tau^C)=0$ and thus $\textbf{1}_{\lbrace{\tau^C< \tau^B\rbrace}}+\textbf{1}_{\lbrace{\tau^B< \tau^C\rbrace}}=1$ $\mathbb{P}$-a.s.
\end{remark}
The cashflow process associated to a defaultable, collateralized contract $(A,R,C,\tau)$, denoted by $A^{R,C}=(A_t^{R,C})_{t \in [0,\tau \wedge T]}$, is given by
\begin{align} \label{eq:CashflowDefaultCollateralized}
    A_t^{R,C}:= {A}_t^{C}+ \textbf{1}_{\lbrace t \geq \tau  \rbrace}R_{\tau},\quad t \in [0,\tau \wedge T],
\end{align}
with 
\begin{align} \label{eq:CashflowDefaultCollateralized1}
    {A}_t^C= \textbf{1}_{\lbrace t < \tau  \rbrace} A_t +  \textbf{1}_{\lbrace t \geq \tau  \rbrace} A_{\tau-}+ \textbf{1}_{\lbrace t < \tau  \rbrace} C_t + \textbf{1}_{\lbrace t \geq \tau  \rbrace}C_{\tau-}-\int_0^{\tau \wedge t} \bar{r}_u^c C_u du,\quad t \in [0,\tau \wedge T],
\end{align}
and the recovery payoff $R_{\tau}$ is defined in \eqref{eq:finalCloseOut}.

\subsection{Portfolio and value process}
We assume that the bank enters in a defaultable, collateralized contract $(A, R, C, \tau)$ at time $0$ and invests over the time horizon $[0, \tau \wedge T]$ all its wealth fully collateralized in the risky assets and the treasury account $B^f$. In particular, the contract $A$ can only be entered at time $0$ and ends at time $\tau \wedge T.$
\begin{definition}
    We denote by $$\varphi=(\xi, \psi, \psi^f)=(\xi_t, \psi_t, \psi^f_t)_{t \in [0,\tau \wedge T]}$$
the \emph{portfolio/strategy} of the bank where
\begin{enumerate}
    \item $\xi=(\xi^1,...,\xi^d)$ is an $\mathbb{R}^d$-valued, $\mathbb{G}$-predictable process, representing the number of basic traded assets $S^1,...,S^d$,
    \item $\psi=(\psi^1,...,\psi^d)$ is an $\mathbb{R}^d$-valued, $\mathbb{G}$-predictable process, representing the shares of the repo accounts $B^1,...,B^d$ and $\psi$ is uniquely determined by $\xi$ due to the repo-trading constraint \eqref{eq:repoConstraint},
    \item $\psi^f$ is an $\mathbb{R}$-valued, $\mathbb{G}$-adapted process denoting the shares in the treasury account $B^f$.
\end{enumerate}
\end{definition}
\begin{definition}
    The \emph{value process} of a portfolio $\varphi$ for a defaultable, collateralized contract $(A,R,C,\tau)$, denoted by $V^p(\varphi,A^{R,C})=(V_t^p(\varphi,A^{R,C}))_{t \in [0,\tau \wedge T]}$, is given by
\begin{equation} \label{eq:ValuePortfolioNoDefault}
V_t^p(\varphi,A^{R,C}):= \psi_t^f B_t^f + \sum_{i=1}^d \left(\psi_t^i B_t^i + \xi_t^i S_t^i \right)=\psi_t^fB_t^f, \quad t \in [0, \tau \wedge T],
\end{equation}
where we use the repo-trading constraint in \eqref{eq:repoConstraint}.
\end{definition}
Note that while the value process $V^p(\varphi, A^{R,C})$ does not depend on the contract $(A,R,C,\tau)$, the dynamics of $V^p(\varphi, A^{R,C})$ do, see also \cite{bbfpr2022}.\\
\\
We now want to hedge a defaultable, collateralized contract $(A,R,C,\tau)$ with a suitable strategy $\varphi$.
\begin{definition} \label{def:ValueProcessWealth}
    The \emph{value process} $V(\varphi,A^{R,C})=(V_t(\varphi,A^{R,C}))_{t \in [0, \tau \wedge T]}$ of the bank's wealth associated to a portfolio $\varphi$ for a defaultable, collateralized contract $(A,R,C,\tau)$ is given on the event $\lbrace{t<\tau \rbrace}$ for every $t \in [0, \tau \wedge T]$ by
    \begin{align*}
        V_t(\varphi,A^{R,C}):=V_t^p(\varphi,A^{R,C})-C_t= V_t^{p}(\varphi,A^C)-C_t,
    \end{align*}
and on the event $\lbrace{\tau \leq T \rbrace}$ by
\begin{align*}
 V_{\tau}(\varphi,A^{R,C}):= V_{\tau}^p(\varphi,A^{R,C})= V_{\tau}^p(\varphi,{A}^{C})+R_{\tau}.
\end{align*}
\end{definition}
\begin{remark} Note that we assumed boundedness for all interest rates for the sake of simplicity in order to guarantee existence and uniqueness of the BSDEs we are going to consider in the sequel. Boundedness can be dropped by resorting to the results of e.g. \cite{elkaroui1997general}, \cite{Yong2006} or \cite{Gashi2019}.
Concerning the deep-learning-based numerical algorithms mentioned in
the Introduction, we remark that boundedness of the interest rates is not necessary to ensure convergence.
 For example, if a short rate is modeled by a
Vasicek model, the corresponding BSDE generator remains Lipschitz in the
solution variables, but the associated Lipschitz coefficient becomes stochastic
rather than uniformly bounded. Well-posedness and continuous dependence for
BSDEs under stochastic Lipschitz conditions are available under suitable
integrability assumptions; see \cite{LiLaiFan2023}. This suggests a natural
localization procedure. More precisely, for $r_t^{x,N}:=(-N)\vee r_t^x\wedge N$,
the BSDE corresponding to the truncated rate $r^{x,N}$ has a uniformly bounded
rate-dependent Lipschitz coefficient and, subject to the remaining assumptions,
falls within the uniformly Lipschitz setting covered by existing convergence
results for deep BSDE schemes, such as \cite{HanLong2020}. The general stability
results of \cite{PapapantoleonPossamaiSaplaouras2023} then provide a framework,
under suitable convergence and integrability conditions, for passing
from the localized BSDEs to the original stochastic-Lipschitz BSDE as
$N\to\infty$. In this way, the numerical error can be separated into the deep
BSDE approximation error for the localized problem and the localization error
controlled by BSDE stability.

\end{remark}

\section{Local risk-minimization in a multi-curve setting}\label{Sec:LocalRiskMultiCurve}
Our goal is to address the market incompleteness due to the presence of possible defaults via local risk-minimization extended to a multi-curve setting. According to this approach, also in the case of an incomplete market a contingent claim or a payment stream is perfectly replicated via a portfolio admitting a cost. 
In the sequel we define the cost process of a financial contract $A$, which is then used to introduce the definition of a locally risk-minimizing strategy and a F\"ollmer-Schweizer decomposition in the presence of multi-curves. We then link the existence of a locally risk-minimizing strategy to a F\"ollmer-Schweizer decomposition.

\subsection{Locally risk-minimizing strategies and F\"ollmer-Schweizer decomposition in a multi-curve setting}
Before providing the definition of a pseudo-locally risk-minimizing strategy in this setting we introduce some further notation.
\begin{notation} \label{notation:SpacesLocalRisk}
\begin{enumerate}
\item {We denote by $\Theta_{X}^{\mathbb{G},\tau}$ the space of all $\mathbb{R}^d$-valued, $\mathbb{G}$-predictable processes $\xi=(\xi_t)_{t \in [0, \tau \wedge T]}$ such that for all $i=1,...,d$ the integral $\int_0^{\cdot \wedge \tau} \xi_u^i \frac{B_u^i}{B_u^f}  d\widetilde{S}_u^i$ is well-defined and in $\mathcal{S}^2_{\textnormal{d}}(\mathbb{G}, \mathbb{P})$. Equivalently, $\Theta_X^{\mathbb{G},\tau}$ contains all $\mathbb{R}^d$-valued, $\mathbb{G}$-predictable processes $\xi=(\xi_t)_{t \in [0,\tau \wedge T]}$ such that
\begin{align} \label{eq:integrabilityTau}
\mathbb{E}_{\mathbb{P}}\left[ \int_0^{\tau\wedge T} \xi^{\top}_u d \langle M^{X}\rangle_u \xi_u + \left( \int_0^{\tau\wedge T} \left \vert \xi^{\top}_u dA^{X}_u\right \vert \right)^2 \right] < \infty.
\end{align}
\item $\Theta_X^{\mathbb{F}, \tau}$ is the space of all $\mathbb{R}^d$-valued, $\mathbb{F}$-predictable processes $\xi=(\xi_t)_{t \in [0,\tau \wedge T]}$ such that \eqref{eq:integrabilityTau} holds.}
\end{enumerate}
\end{notation}
Locally risk-minimizing strategies are not self-financing. However, we rely on Definition 7 in \cite{bbfpr2022} of a self-financing strategy to define the cost process. 
\begin{definition}
    We say that a trading strategy $\varphi=(\xi, \psi,\psi^f)$ with $\xi \in \Theta_X^{\mathbb{G},\tau}$ is a \emph{self-financing strategy} inclusive of a defaultable, collateralized contract $(A,R,C, \tau)$, if the value process $V^p_t(\varphi,A^{R,C})=\psi_t^f B_t^f$ satisfies on $[0,\tau \wedge T]$ 
    \begin{align} \label{eq:self-financing}
        d V_t^p(\varphi,A^{R,C})&= \psi^f_t dB_t^f + \sum_{i=1}^d \left( \psi_t^i dB_t^i + \xi_t^i dS_t^i\right)+dA_t^{R,C}.
    \end{align}
\end{definition}

For a self-financing strategy $\varphi$ inclusive of a defaultable, collateralized contract $(A,R,C, \tau)$ it holds
\begin{align}
d V_t^p(\varphi,A^{R,C})&= \psi^f_t dB_t^f + \sum_{i=1}^d \left( \psi_t^i dB_t^i + \xi_t^i dS_t^i\right)+dA_t^{R,C} \label{eq:CalculationStart} \\
    &=\psi^f_t dB_t^f - \sum_{i=1}^d  \xi_t^i S_t^i (B_t^i)^{-1} dB_t^i +  \sum_{i=1}^d \xi_t^i dS_t^i+dA_t^{R,C} \label{eq:RepoConstraint1}, \\
    &=\psi^f_t B_t^f \frac{dB_t^f}{B_t^f} + \sum_{i=1}^d  \xi_t^i B_t^i \frac{1}{B_t^i}\left( dS_t^i -r_t^i S_t^i dt \right) + dA_t^{R,C} \label{eq:DynamicsBank},  \\
    &=V_t^p(\varphi,A^{R,C})r_t^f dt +\sum_{i=1}^d  \xi_t^i B_t^i d \widetilde{S}_{t}^{i} + dA_t^{R,C}, \label{eq:DynamicsValuePortfolioNoDefaultNoCollateral}
    %& =V_t^p(\varphi,A)r_t dt + V_t^p(\varphi,A)\left(r_t^f-r_t \right) dt+\sum_{i=1}^d  \xi_t^i B_t^i d \widetilde{S}_{t}^{i} + dA_t, 
\end{align}
where we use the repo-trading constraint $\psi_t^i=-(B_t^{i})^{-1}\xi_t^i S_t^i$ in \eqref{eq:RepoConstraint1}. Moreover, \eqref{eq:DynamicsBank} follows by \eqref{eq:DyS}. In \eqref{eq:DynamicsValuePortfolioNoDefaultNoCollateral} we use that $V^p(\varphi, A^{R,C})$ satisfies \eqref{eq:ValuePortfolioNoDefault} and $dB_t^f=r_t^fB_t^fdt$. By setting 
\begin{align} \label{eq:DiscountedValueProcess}
\widetilde{V}_t^p (\varphi,A^{R,C}):= \frac{V_t^p(\varphi,A^{R,C})}{B_t^f}, \quad t \in [0,\tau \wedge T],
\end{align}
we conclude with \eqref{eq:DynamicsValuePortfolioNoDefaultNoCollateral}
\begin{align}
   d \widetilde{V}_t^p (\varphi,A^{R,C})&=\frac{1}{B_t^f}dV_t^p(\varphi,  A^{R,C})-\frac{V_t^p(\varphi,  A^{R,C})}{B_t^f}r_t^f dt \nonumber \\
   &=\frac{1}{B_t^f} \left( V_t^p(\varphi,A^{R,C})r_t^f dt +\sum_{i=1}^d  \xi_t^i B_t^i d \widetilde{S}_{t}^{i} + dA_t^{R,C} \right) -\frac{V_t^p(\varphi,  A^{R,C})}{B_t^f}r_t^f dt \nonumber \\
   & = \frac{1}{B_t^f} \left( \sum_{i=1}^d  \xi_t^i B_t^i d \widetilde{S}_{t}^{i} + dA_t^{R,C} \right).
\end{align}
This motivates the following definition of the cost process of a portfolio.
\begin{definition} \label{def:CostProcessWithDefault} Let $(\int_0^t \frac{1}{B_u^f}d A_u^{R,C})_{t \in [0, \tau \wedge T]} \in L^2(\mathbb{G}, \mathbb{P}).$
The \emph{cost process of a portfolio} $\varphi=(\xi, \psi, \psi^f)$ with $\xi \in \Theta_X^{\mathbb{G}, \tau}$for a defaultable, collateralized contract $(A,R,C,\tau)$, denoted by $\mathcal{C}^{\varphi,A^{R,C}}=(\mathcal{C}_t^{\varphi,A^{R,C}})_{t \in [0,\tau \wedge T]}$, is given by
\begin{align} \label{eq:CostPocessPortfolioWithDefault}
    \mathcal{C}_t^{\varphi,A^{R,C}}:= - \int_{0}^t \frac{1}{B_u^f} dA_u^{R,C} + \widetilde{V}_t^p (\varphi,A^{R,C})  - \sum_{i=1}^d  \int_{0}^t  \xi_u^i \frac{B_u^i}{B_u^f} d \widetilde{S}_{u}^{\tau,i}, \quad t \in [0,\tau \wedge T].
\end{align}
\end{definition}
In the previous definition we make use of a slight abuse of notation by stating $(\int_0^t \frac{1}{B_u^f}d A_u^{R,C})_{t \in [0, \tau \wedge T]} \in L^2(\mathbb{G}, \mathbb{P})$, as processes in $L^2(\mathbb{G}, \mathbb{P})$ are defined on $[0,T]$ and not on $[0, \tau \wedge T]$. 
% To simplify the notation we use this convention also for the other spaces introduced in Notation \ref{notation:SpacesLocalRisk0}.

\begin{remark} \label{remark:IntegralStopped}
\begin{enumerate}
\item  Note that $\xi \in \Theta_X^{\mathbb{G},\tau}$ guarantees that the stochastic integrals in \eqref{eq:CostPocessPortfolioWithDefault} with respect to $\widetilde{S}^{\tau,i}$ are well-defined.
 \item By Theorem 3 in Chapter VIII in \cite{dm82} for $\xi \in \Theta_X^{\mathbb{F},\tau}$ and $\xi \in \Theta_X^{\mathbb{G},\tau}$ it holds
\begin{equation} \label{eq:IntegralStopping}
\int_0^{t \wedge \tau} \xi_u^i \frac{B_u^i}{B_u^f} d\widetilde{S}_u^i =  \int_0^t \textbf{1}_{\lbrace u \leq \tau \rbrace}\xi_u^i \frac{B_u^i}{B_u^f} d\widetilde{S}_u^{i}= \int_0^t \xi_u^i \frac{B_u^i}{B_u^f} d\widetilde{S}_u^{\tau,i}, \quad t \in [0, T], \, i=1,...,d.
\end{equation}
\end{enumerate}
\end{remark}
 In general, we assume that the bank invests in the risky assets, repo accounts and the treasury account according to the information available on the financial market and the occurrence of the defaults, i.e. $\xi$ is $\mathbb{G}$-predictable and $\psi^f$ is $\mathbb{G}$-adapted. 
\begin{definition} \label{def:L2StrategyGeneralSetting}
\begin{enumerate}
\item An \emph{$L^2_{\mathbb{G}}$-strategy} for a defaultable, collateralized contract $(A,R,C,\tau)$ is a portfolio $\varphi=(\xi, \psi, \psi^f)$ such that  $\xi \in \Theta_{X}^{\mathbb{G},\tau}$, $\psi$ determined by the repo-trading constraint \eqref{eq:repoConstraint} and $\psi^f$ is an $\mathbb{R}$-valued, $\mathbb{G}$-adapted process such that the discounted value process $\widetilde{V}(\varphi,A^{R,C})$  of the bank's wealth in \eqref{eq:DiscountedValueProcess} is right-continuous and square-integrable on $[0, \tau\wedge T].$
      \item An $L^2_{\mathbb{G}}$-strategy for a defaultable, collateralized contract $(A,R,C,\tau)$, denoted by $\varphi$, is called \emph{pseudo-locally risk-minimizing}, if $\varphi$ 
      \begin{itemize}
      \item is \emph{$0$-achieving}, i.e. $\widetilde{V}_{\tau \wedge T}(\varphi,A^{R,C})=0$ $\mathbb{P}$-a.s.,
      \item is \emph{mean-self-financing} for $(A,R,C,\tau)$, i.e. the cost process $\mathcal{C}^{\varphi, A^{R,C}}$ in \eqref{eq:CostPocessPortfolioWithDefault} is a $(\mathbb{G},\mathbb{P})$-martingale, (which is then also square-integrable),
      \item $\mathcal{C}^{\varphi, A^{R,C}}$ in \eqref{eq:CostPocessPortfolioWithDefault} is strongly orthogonal\footnote{i.e. for all $i=1,...,d$ the quadratic variation process $\left \langle \mathcal{C}^{\varphi, A^{R,C}}, M^{X^{\tau, i}} \right \rangle$ is equal to $0$ $\mathbb{P}$-a.s. Note that $X^{\tau} \in \mathbb{S}^2_{\textnormal{d}}(\mathbb{G}, \mathbb{P})$.} to the martingale part of $X^{\tau}$.
      \end{itemize}
      \item Analogously, we define a pseudo-locally risk-minimizing $L_{\mathbb{F}}^2$-strategy $\varphi=(\xi, \psi, \psi^f)$ by substituting $\xi\in \Theta_X^{\mathbb{G}, \tau}$  by $\xi \in \Theta_X^{\mathbb{F}, \tau}$ in the definitions.
\end{enumerate}
\end{definition}
\begin{remark}
By Theorem 1.6 in \cite{schweizer_survey} the definitions of a pseudo-locally risk-minimizing $L^2_{\mathbb{G}}$-strategy and of a locally risk-minimizing $L^2_{\mathbb{G}}$-strategy are equivalent, see Definition 1.5 in \cite{schweizer_survey}, if $X^{\tau}$ is continuous, satisfies the structure condition and the associated mean-variance tradeoff process is also continuous, see Appendix \ref{appendix}. In particular, as these conditions are satisfied in our setting, see Proposition \ref{prop:Existence}, we refer to a pseudo-locally risk-minimizing $L^2_{\mathbb{G}}$-strategy as locally risk-minimizing, to avoid the introduction of the very technical definition of the latter one. 
\end{remark}
We introduce the notation $\mathcal{L}=(\mathcal{L}_t)_{t \in [0,\tau \wedge T]}$ for the payment stream
\begin{align}\label{eq:PaymentStreamL}
\mathcal{L}_t:= \int_0^t \frac{1}{B_u^f}dA_u^{R,C} - \frac{C_t}{B_t^f}\textbf{1}_{\lbrace t < \tau \rbrace} + \frac{R_{\tau}}{B_{\tau}^f}\textbf{1}_{\lbrace t\geq \tau  \rbrace  },
\end{align}
and assume from now on that $\mathcal{L} \in L^2(\mathbb{G},\mathbb{P}).$

Our aim is to link the existence of a locally risk-minimizing $L^2_{\mathbb{G}}$-strategy for a defaultable collateralized contract $(A,R,C,\tau)$ to the existence of a (stopped) F\"ollmer-Schweizer decomposition of the random variable $\mathcal{L}_{\tau \wedge T}$ in \eqref{eq:PaymentStreamL}.
In the following, we denote by $\mathcal{M}_0^2(\mathbb{G}, \mathbb{P})$ the space of all right-continuous square-integrable $(\mathbb{G}, \mathbb{P})$-martingales $M=(M_t)_{t \in [0, \tau \wedge T]}$ with $M_0=0.$
\begin{definition} \label{def:FoellmerSchweizerWithDefault}
The random variable $\mathcal{L}_{\tau \wedge T}$ in \eqref{eq:PaymentStreamL} admits a \emph{(stopped) $\mathbb{G}$-F\"ollmer-Schweizer decomposition} if it can be written as
    \begin{align} \label{eq:FoellmerSchweizerWithDefault}
     -\mathcal{L}_{\tau \wedge T} = h_0 + \sum_{i=1}^d  \int_{0}^{T}  {\xi}_u^{\mathcal{L},i} \frac{B_u^i}{B_u^f} d \widetilde{S}_{u}^{\tau,i} + {H}_{\tau \wedge T}^{\mathcal{L}} \quad \mathbb{P}\text{-a.s.,}
    \end{align}
    where $h_0 \in \mathbb{R}$, $\xi^{\mathcal{L}} \in \Theta_{X}^{\mathbb{G},\tau},$ ${H}^{\mathcal{L}}=({H}_t^{\mathcal{L}})_{t \in [0,\tau \wedge T]} \in \mathcal{M}_0^2(\mathbb{G}, \mathbb{P})$ and is strongly orthogonal to the martingale component of $X^{\tau}.$ If we find a decomposition of $\mathcal{L}_{\tau \wedge T}$ as in \eqref{eq:FoellmerSchweizerWithDefault} with $\xi^{\mathcal{L}} \in \Theta_{X}^{\mathbb{F},\tau}$, then it is called a \emph{(stopped) $\mathbb{F}$-F\"ollmer-Schweizer decomposition}.
\end{definition}

\begin{proposition} \label{prop:EquivalenceLRMFS}
    The defaultable, collateralized contract $(A,R,C,\tau)$ admits a locally risk-minimizing $L^2_{\mathbb{G}}$-strategy $\varphi=(\xi,\psi, \psi^f)$ if and only if the random variable $\mathcal{L}_{\tau \wedge T}$ in \eqref{eq:PaymentStreamL} admits a $\mathbb{G}$-F\"ollmer-Schweizer decomposition as in Definition \ref{def:FoellmerSchweizerWithDefault}. In that case $\varphi=(\xi,\psi, \psi^f)$ is given by
    \begin{align}
         &\xi_t=\xi^{\mathcal{L}}_t, \quad \quad \psi^i_t=- \frac{\xi^{\mathcal{L},i}_t S^i_t}{B^i_t}, \quad t \in [0,\tau \wedge T], \, i=1,...,d, \label{eq:SpecificationLRMS1}\\
        &\psi^f_t= h_0 + \sum_{i=1}^d  \int_{0}^{t}  {\xi}_u^{\mathcal{L},i} \frac{B_u^i}{B_u^f} d \widetilde{S}_{u}^{\tau,i} + {H}_{t}^{\mathcal{L}} + \int_0^{t} \frac{1}{B_u^f} dA_u^{R,C}, \quad t \in [0,\tau \wedge T].\label{eq:SpecificationLRMS2}
    \end{align}
\end{proposition}
\begin{proof}  
Suppose that the random variable $\mathcal{L}_{\tau \wedge T}$ admits a $\mathbb{G}$-F\"ollmer-Schweizer decomposition as in Definition \ref{def:FoellmerSchweizerWithDefault}. This means that it can be written as
\begin{align} \label{eq:FoellmerSchweizerWithDefaultProof}
       - \mathcal{L}_{\tau \wedge  T}= -\int_0^{\tau \wedge T} \frac{1}{B_u^f} dA_u^{R,C} + \frac{C_{T}}{B_{T}^f} \textbf{1}_{\lbrace \tau >T \rbrace  } - \frac{R_{\tau}}{B_{\tau}^f} \textbf{1}_{\lbrace \tau \leq T \rbrace} = h_0 + \sum_{i=1}^d  \int_{0}^{T}  {\xi}_u^{\mathcal{L},i} \frac{B_u^i}{B_u^f} d \widetilde{S}_{u}^{\tau,i} + {H}_{\tau \wedge T}^{\mathcal{L}} \quad \mathbb{P}\text{-a.s.,}
    \end{align}
with $h_0 \in \mathbb{R}$, ${\xi}^{\mathcal{L}} \in \Theta_{X}^{\mathbb{G},\tau},$ ${H}^{\mathcal{L}}=({H}_t^{\mathcal{L}})_{t \in [0,\tau \wedge T]} \in \mathcal{M}_0^2(\mathbb{G}, \mathbb{P})$ and strongly orthogonal to the martingale component of $X^{\tau}.$
We now construct a locally risk-minimizing $L^2_{\mathbb{G}}$-strategy $\varphi=(\xi, \psi, \psi^f)$. We set $\xi:={\xi}^{\mathcal{L}}$ and 
by the repo-trading constraint \eqref{eq:repoConstraint} it holds for $\psi=(\psi^1,..., \psi^d)$ 
$$
\psi^i := - \frac{\xi^i S^i}{B^i}, \quad i=1,...,d.
$$
Moreover, we define
\begin{align} \label{eq:conditionBankAccountWithCollateralDefault}
  \psi_t^f:=h_0 + \sum_{i=1}^d  \int_{0}^t  {\xi}_u^i \frac{B_u^i}{B_u^f} d \widetilde{S}_{u}^{\tau,i} + {H}_{t }^{\mathcal{L}} + \int_0^t \frac{1}{B_u^f} dA_u^{R,C}, \quad t \in [0, \tau \wedge T],
\end{align}
which is $\mathcal{G}_t$-adapted.  
Obviously, it holds $\xi \in \Theta_{X}^{\mathbb{G},\tau}$. {With this choice of $\varphi=(\xi, \psi, \psi^f)$ by \eqref{eq:ValuePortfolioNoDefault} and Definition \ref{def:ValueProcessWealth} it holds $\psi_t^f=\widetilde{V}_t^p(\varphi, A^{R,C})=\widetilde{V}_t(\varphi, A^{R,C})+ \frac{C_t}{B_t^f}\textbf{1}_{\lbrace t < \tau \rbrace} -\frac{R_{\tau}}{B_{\tau}^f} \textbf{1}_{\lbrace t \geq \tau  \rbrace}$ and thus 
\begin{align}
\widetilde{V}_t(\varphi, A^{R,C})&=\psi_t^f -\frac{C_t}{B_t^f} \textbf{1}_{\lbrace t < \tau \rbrace}+ \frac{R_{\tau}}{B_{\tau}^f}\textbf{1}_{\lbrace t \geq \tau \rbrace} \nonumber \\
&=h_0 + \sum_{i=1}^d  \int_{0}^t  {\xi}_u^i \frac{B_u^i}{B_u^f} d \widetilde{S}_{u}^{\tau,i} + {H}_{t }^{\mathcal{L}} + \int_0^t \frac{1}{B_u^f} dA_u^{R,C} - \frac{C_t}{B_t^f} \textbf{1}_{\lbrace t < \tau \rbrace}+ \frac{R_{\tau}}{B_{\tau}^f}\textbf{1}_{\lbrace t \geq \tau \rbrace} \nonumber \\
&= h_0 + \sum_{i=1}^d  \int_{0}^t  {\xi}_u^i \frac{B_u^i}{B_u^f} d \widetilde{S}_{u}^{\tau,i} + {H}_{t }^{\mathcal{L}}+ \mathcal{L}_t, \nonumber
\end{align}
which is again square-integrable and right-continuous as ${\xi} \in \Theta_{X}^{\mathbb{G},\tau},$ ${H}^{\mathcal{L}}=({H}_t^{\mathcal{L}})_{t \in [0,\tau \wedge T]} \in \mathcal{M}_0^2(\mathbb{G}, \mathbb{P})$ and $(\mathcal{L}_t)_{t \in [0,\tau \wedge T]} \in L^2(\mathbb{G}, \mathbb{P}).$ This allows to conclude that $\varphi$ is a $L^2_{\mathbb{G}}$-strategy.
}
The corresponding cost process $C^{\varphi,A^{R,C}}$ of the strategy $\varphi=(\xi, \psi, \psi^f)$ is given for $t \in [0, \tau \wedge T]$ by
    \begin{align} 
    \mathcal{C}_t^{\varphi,A^{R,C}}&= - \int_{0}^t \frac{1}{B_u^f} dA_u^{R,C} + \widetilde{V}_t^p (\varphi,A^{R,C})  - \sum_{i=1}^d  \int_{0}^t  \xi_u^i \frac{B_u^i}{B_u^f} d \widetilde{S}_{u}^{\tau,i} \nonumber \\
    &= - \int_{0}^t \frac{1}{B_u^f} dA_u^{R,C} + h_0 + \sum_{i=1}^d  \int_{0}^t  {\xi}_u^i \frac{B_u^i}{B_u^f} d \widetilde{S}_{u}^{\tau,i} + {H}_t^{\mathcal{L}} + \int_0^t \frac{1}{B_u^f} dA_u^{R,C}  - \sum_{i=1}^d  \int_{0}^t  \xi_u^i \frac{B_u^i}{B_u^f} d \widetilde{S}_{u}^{\tau,i} \nonumber \\
     &=  h_0 +  {H}_t^{\mathcal{L}}, \nonumber 
\end{align}
from \eqref{eq:CostPocessPortfolioWithDefault}. 
Thus, by the definition of the process $H^{\mathcal{L}}$ in the $\mathbb{G}$-F\"ollmer-Schweizer decomposition the cost process $\mathcal{C}^{\varphi,A^{R,C}}$ belongs to $ \mathcal{M}_0^2(\mathbb{G}, \mathbb{P})$ and is strongly orthogonal to the martingale components of $X^{\tau}$. Hence, $\varphi$ is mean-self-financing. Moreover, by combining \eqref{eq:FoellmerSchweizerWithDefaultProof} and \eqref{eq:conditionBankAccountWithCollateralDefault} it follows
\begin{align*} 
    \widetilde{V}_{\tau \wedge T}^p(\varphi, A^{R,C})&=h_0 + \sum_{i=1}^d  \int_{0}^{T}  {\xi}_u^i \frac{B_u^i}{B_u^f} d \widetilde{S}_{u}^{\tau,i} + {H}_{\tau \wedge T}^{\mathcal{L}}  + \int_0^{\tau \wedge T}  \frac{1}{B_u^f} dA_u^{R,C}= \frac{C_{T}}{B_{T}^f} \textbf{1}_{\lbrace \tau >T \rbrace  } - \frac{R_{\tau}}{B_{\tau}^f} \textbf{1}_{\lbrace \tau \leq T \rbrace} \quad \mathbb{P}\text{-a.s.,}
\end{align*}
i.e.  $\varphi$ is $0$-achieving. 
\\
Conversely, let $\varphi$ be a locally risk-minimizing $L^2_{\mathbb{G}}$-strategy for the defaultable, collateralized contract $(A,R,C, \tau)$. As $\varphi$ is $0$-achieving it holds $\widetilde{V}_{\tau \wedge T}^p(\varphi, A^{R,C})=\frac{C_{T}}{B_{T}^f} \textbf{1}_{\lbrace \tau >T \rbrace  } - \frac{R_{\tau}}{B_{\tau}^f} \textbf{1}_{\lbrace \tau \leq T \rbrace}$. By using the definition of $\mathcal{C}^{\varphi,A^{R,C}}$ in \eqref{eq:CostPocessPortfolioWithDefault} we get
\begin{align*}
0&=\widetilde{V}_{\tau \wedge T}^p(\varphi,A^{R,C})-\frac{C_{T}}{B_{T}^f} \textbf{1}_{\lbrace \tau >T \rbrace  } + \frac{R_{\tau}}{B_{\tau}^f} \textbf{1}_{\lbrace \tau \leq T \rbrace} \\
&= \mathcal{C}_{\tau \wedge T}^{\varphi,A^{R,C}}+ \int_{0}^{\tau \wedge T} \frac{1}{B_u^f} dA_u^{R,C} +  \sum_{i=1}^d  \int_{0}^{ T}  \xi_u^i \frac{B_u^i}{B_u^f} d \widetilde{S}_{u}^{\tau,i} -\frac{C_{T}}{B_{T}^f} \textbf{1}_{\lbrace \tau >T \rbrace  } + \frac{R_{\tau}}{B_{\tau}^f} \textbf{1}_{\lbrace \tau \leq T \rbrace} \quad \mathbb{P}\text{-a.s,}
\end{align*}
or equivalently, 
\begin{align*}
- \mathcal{L}_{\tau \wedge T} = \mathcal{C}_0^{\varphi,A^{R,C}} + \left(\mathcal{C}_{\tau \wedge T}^{\varphi,A^{R,C}}-\mathcal{C}_0^{\varphi,A^{R,C}}\right)+ \sum_{i=1}^d  \int_{0}^{ T}  \xi_u^i \frac{B_u^i}{B_u^f} d \widetilde{S}_{u}^{\tau,i}  \quad \mathbb{P}\text{-a.s.}
\end{align*}
Thus choosing
$$
\xi^{\mathcal{L}}_t:= \xi_t, \quad h_0:=\mathcal{C}_0^{\varphi,A^{R,C}} \quad \text{and}\quad {H}^{\mathcal{L}}_{t}:=\mathcal{C}_{\tau \wedge t}^{\varphi,A^{R,C}}-\mathcal{C}_0^{\varphi,A^{R,C}}, \quad t \in [0,\tau \wedge T],
$$
provides a $\mathbb{G}$-F\"ollmer-Schweizer decomposition.
\end{proof}
We now prove an existence result for a (stopped) $\mathbb{G}$-Föllmer-Schweizer decomposition. For the reader's convenience we provide the definition of the so-called structure condition and the mean-variance tradeoff process in Appendix \ref{theorem:Schweizer}.
\begin{proposition} \label{prop:Existence}
Let $X^{\tau}=(X^{\tau}_t)_{t \in [0, \tau \wedge T]}$ and $\mathcal{L}=(\mathcal{L}_t)_{t \in [0,\tau \wedge T]}$ be as in \eqref{eq:SemimartingaleDecomposotionX} and \eqref{eq:PaymentStreamL}, respectively. 
Then, $X^{\tau}$ satisfies the structure condition and the associated mean-variance tradeoff process $K^{X^{\tau}}=(K^{X^{\tau}}_t)_{t \in [0, \tau \wedge T]}$ is uniformly bounded and continuous. Thus, there exists a unique $\mathbb{G}$-F\"ollmer-Schweizer decomposition for $\mathcal{L}_{\tau \wedge T}$ and a unique locally risk-minimizing $L_{\mathbb{G}}^2$-strategy $\varphi=(\xi, \psi, \psi^f)$ for the defaultable, collateralized contract $(A,R,C,\tau)$.
\end{proposition}
\begin{proof}
For $i=1,...,d$ let $M^{X^{\tau,i}}$ and $A^{X^{\tau,i}}$ be given as in \eqref{eq:SemimartingaleDecomposotionX}.
Note that $\left \langle M^{X^{\tau,i}}, M^{X^{\tau,j}} \right\rangle_t=0$ for $i \neq j$, $t \in [0, \tau \wedge T]$ and 
$$ 
d\left \langle M^{X^{\tau,i}}\right \rangle_t=\left(\frac{S_t^{\tau,i}}{B_t^f}\right)^2 \left({\sigma}^{i}(t,S_t^{\tau})\right)^2 dt \quad \text{for }t \in [0, \tau \wedge T].
$$
Defining the process $\lambda^i=(\lambda^i_u)_{u \in [0, \tau \wedge T]}$ by
\begin{align}
    \lambda_u^i={\frac{\mu^i(u,S_u^{\tau}) -r_u^i}{({\sigma}^{i}(u,S_u^{\tau}))^2\frac{S_u^{\tau,i}}{B_u^f}} }, \nonumber 
\end{align}
it follows for $t \in[0,\tau \wedge T]$ that
$$
A_t^{X^{\tau,i}}=\int_0^t  {\frac{\mu^i(u,S_u^{\tau}) -r_u^i}{({\sigma}^{i}(u,S_u^{\tau}))^2\frac{S_u^{\tau,i}}{B_u^f}} }
\left(\frac{S_u^{\tau,i}}{B_u^f}\right)^2 \left({\sigma}^{i}(u,S_u^{\tau})\right)^2  du= \int_0^t \lambda_u^i d \left \langle M^{X^{\tau,i}} \right \rangle_u.
$$
The mean-variance tradeoff $K^{X^{\tau}}$ is uniformly bounded in $t$ and $\omega$, as for $t \in[0,\tau \wedge T]$
\begin{align*}
K_t^{X^{\tau}}= \sum_{i=1}^d \int_0^{t} \left(\lambda_u^i\right)^2 \left(\frac{S_u^{\tau,i}}{B_u^f}\right)^2 \left({\sigma}^{i}(u,S_u^{\tau})\right)^2 du=\sum_{i=1}^d \int_0^t \left(\frac{\mu^i(u,S_u^{\tau}) -r_u^i}{{\sigma}^{i}(u,S_u^{\tau})}\right)^2 du  \leq d \left(\max_{i=1,...,d}K_i\right)^2T, 
\end{align*}
by using \eqref{eq:BoundsExistenceSolutionDefaults}. Thus, $X^{\tau}$ satisfies the structure condition, see Definition \ref{def:MeanVariance}.\\\
As $X^{\tau}$ is continuous, satisfies the structure condition and $K^{X^{\tau}}$ is uniformly bounded, there exists a unique $\mathbb{G}$-F\"ollmer-Schweizer decomposition of $\mathcal{L}_{\tau \wedge T}$ with respect to $X^{\tau}$ by Theorem \ref{theorem:Schweizer}. The existence and uniqueness of a locally risk-minimizing strategy follows immediately by Proposition \ref{prop:EquivalenceLRMFS}. 
\end{proof}
 We now briefly discuss the case if we allow the bank to invest in the risky assets and the respective repo accounts, based on the information from the financial market \emph{before} the first default, i.e. $\xi$ is $\mathbb{F}$-predictable and we allow to adjust with $\psi^f,$ which is $\mathbb{G}$-adapted. 

\begin{remark} \label{remark:FoellmerSchweizerDifferentFiltration}
The existence of a $\mathbb{G}$-F\"ollmer-Schweizer decomposition of $\mathcal{L}_{\tau \wedge T}$ implies the existence of an $\mathbb{F}$-F\"ollmer-Schweizer decomposition of $\mathcal{L}_{\tau \wedge T}$. 
Let $\xi^{\mathcal{L}}\in \Theta_{X}^{\mathbb{G},\tau}$ as in Definition \ref{def:FoellmerSchweizerWithDefault}. Then, there exists an $\mathbb{F}$-predictable process $\overline{\xi}^{\mathcal{L}}$ such that 
 \begin{align} \label{eq:PreDefault}
  \textbf{1}_{\lbrace t \leq \tau \rbrace} \overline{\xi}^{\mathcal{L}}_t=\textbf{1}_{\lbrace t \leq \tau \rbrace} {\xi}_t^{\mathcal{L}}
\end{align}
holds, see e.g. \cite{bielecki_jeanblanc_rutkowski}. Therefore, we can rewrite a $\mathbb{G}$-F\"ollmer-Schweizer decomposition of $\mathcal{L}_{\tau \wedge T}$ with respect to ${\xi}^{\mathcal{L}}$ as follows
\begin{align} 
        -\mathcal{L}_{\tau \wedge T} &= h_0 + \sum_{i=1}^d  \int_{0}^{T}  {\xi}_u^{\mathcal{L},i} \frac{B_u^i}{B_u^f} d \widetilde{S}_{u}^{\tau,i} + {H}_{\tau \wedge T}^{\mathcal{L}}, \quad \mathbb{P}\text{-a.s.,} \nonumber \\ 
        &= h_0 + \sum_{i=1}^d  \int_{0}^{T}  \textbf{1}_{\lbrace u \leq \tau \rbrace}{\xi}_u^{\mathcal{L},i} \frac{B_u^i}{B_u^f} d \widetilde{S}_{u}^{i} + {H}_{\tau \wedge T}^{\mathcal{L}}, \quad \mathbb{P}\text{-a.s.,} \label{eq:IntegralStopped1}\\
         &= h_0 + \sum_{i=1}^d  \int_{0}^{T}  \textbf{1}_{\lbrace u \leq \tau \rbrace}\overline{\xi}_u^{\mathcal{L},i} \frac{B_u^i}{B_u^f} d \widetilde{S}_{u}^{i} + {H}_{\tau \wedge T}^{\mathcal{L}}, \quad \mathbb{P}\text{-a.s.,} \label{eq:Predefault1} \\
         &= h_0 + \sum_{i=1}^d  \int_{0}^{T}  \overline{\xi}_u^{\mathcal{L},i} \frac{B_u^i}{B_u^f} d \widetilde{S}_{u}^{\tau,i} + {H}_{\tau \wedge T}^{\mathcal{L}}, \quad \mathbb{P}\text{-a.s.,} \label{eq:IntegralStopped2}
    \end{align}
    where we use in \eqref{eq:IntegralStopped1} and \eqref{eq:IntegralStopped2} the equality \eqref{eq:IntegralStopping} and \eqref{eq:Predefault1} follows by \eqref{eq:PreDefault}. As $\overline{\xi}^{\mathcal{L}}$ satisfies the integrability condition in \eqref{eq:integrabilityTau}, it follows that $\overline{\xi}^{\mathcal{L}} \in \Theta_{X}^{\mathbb{F},\tau}$, and the $\mathbb{F}$- and $\mathbb{G}$-F\"ollmer-Schweizer decompositions coincide. 
  On the other hand, if there exists an $\mathbb{F}$-F\"ollmer-Schweizer decomposition of $\mathcal{L}_{\tau \wedge T}$ with ${\xi}^{\mathcal{L}} \in \Theta_{X}^{\mathbb{F},\tau},$ then $\xi^{\mathcal{L}}$ is also an element in $\Theta_{X}^{\mathbb{G},\tau}$, and provides a $\mathbb{G}$-F\"ollmer-Schweizer decomposition of $\mathcal{L}_{\tau \wedge T}$.
  Thus, $\mathcal{L}_{\tau \wedge T}$ admits a $\mathbb{G}$-F\"ollmer-Schweizer decomposition if and only if it admits an $\mathbb{F}$-F\"ollmer-Schweizer decomposition, and in this case the decompositions coincide.\\
  By using similar arguments it also follows that the defaultable, collateralized contract $(A,R,C,\tau)$ admits a locally risk-minimizing $L^2_{\mathbb{F}}$-strategy $\varphi$ if and only if the random variable $\mathcal{L}_{\tau \wedge T}$ admits an $\mathbb{F}$-F\"ollmer-Schweizer decomposition as in Definition \ref{def:FoellmerSchweizerWithDefault}.
\end{remark}

\section{Characterization of the F\"ollmer-Schweizer decomposition via a BSDE}\label{Sec:FSdecBSDE}
\allowdisplaybreaks
In this section we derive a BSDE whose unique solution provides components of the $\mathbb{G}$-F\"ollmer-Schweizer decomposition of $\mathcal{L}_{\tau \wedge T}$.
To do so, we introduce some notation in order to define the solution of a BSDE.
\begin{notation}  
\begin{enumerate}
 \item We denote by $\cH^{2,d}_{T}(\mathbb{F},\PP)$ the subspace of all $\RR^d$-valued, $\FF$-adapted processes $X$ such that  
	\begin{align} \nonumber 
		\left\| X \right\|_{\cH^{2,d}_{T}}:=\sqrt{\mathbb{E}_{\PP}\left[\int_0^T \left|X_t\right|^2 dt\right]}<\infty. 
	\end{align}
 We set $\cH^{2,d}(\mathbb{F},\PP) := \cH^{2,d}_{T}(\mathbb{F},\PP)$.
 \item The subspace of all $\RR^d$-valued, $\FF$-adapted processes $X=(X_t)_{t \in [0,T]}$ such that 
	\begin{align} \nonumber 
		\left\| X \right\|_{\mathbb{S}^{2,d}_{T}}:=\sqrt{\EE_{\PP}\left[\sup_{t \,\in\, [0,T]}  \left|X_t\right|^2 \right]} < \infty 
	\end{align}
	is denoted by $\mathbb{S}^{2,d}_{T}(\FF,\PP).$ We set $\mathbb{S}^{2}(\FF,\PP) := \mathbb{S}^{2,1}_{T}(\FF,\PP)$.
 \item We denote by $\cH^{2,2}_{\lambda}(\FF, \PP)$ the space of $\FF$-adapted processes $X=(X_t)_{t \in [0,T]}$ with values in $\mathbb{R}^2$ such that 
\begin{align} \nonumber
	\left\| X \right\|_{\cH^{2,2}_{\lambda}}:= \sqrt{\EE_{\PP}\left[\int_0^T\left|X^1_t\right|^2 \lambda^{B}_t+\left|X^2_t\right|^2 \lambda^{C}_t dt \right]} < \infty .
\end{align}
\end{enumerate}
We use similar notations with respect to the filtration $\mathbb{G}$ and the upper integral boundary $\tau \wedge T.$
\end{notation}
\begin{theorem} \label{theorem:LinkBSDE}
Let $X^{\tau}=(X^{\tau}_t)_{t \in [0, \tau \wedge T]}$ and $\mathcal{L}=(\mathcal{L}_t)_{t \in [0,\tau \wedge T]}$ be as in \eqref{eq:SemimartingaleDecomposotionX} and \eqref{eq:PaymentStreamL}, respectively.
Then the $\mathbb{G}$-F\"ollmer-Schweizer decomposition of $\mathcal{L}_{\tau \wedge T} $ 
\begin{align}
   - \mathcal{L}_{\tau \wedge T} 
   &= Y_0 + \sum_{i=1}^d \int_0^T \xi^{\mathcal{L},i}_u \frac{B_u^i}{B_u^f} d\widetilde{S}_u^{\tau,i} + H_{\tau \wedge T}^{\mathcal{L}}, \label{eq:FoellmerSchweizerBSDEtheorem}
\end{align}
is given by 
\begin{align} \label{eq:SepcificationFSXi}
\xi^{\mathcal{L}, i}_t:= \frac{Z_t^{i}B_t^f}{S_t^{i}\sigma^i(t, S_t)} \quad \text{and} \quad \mathcal{H}_t^{\mathcal{L}}:= \sum_{j \in \lbrace B,C  \rbrace}\int_0^{t} U_u^j dM_u^j, \quad t \in [0, \tau \wedge T],
\end{align}
where {$(Y,Z,U^C, U^B) \in \mathbb{S}^2(\mathbb{G}, \mathbb{P}) \times \mathcal{H}^{2,d}(\mathbb{G},\mathbb{P}) \times \mathcal{H}^{2,2}_{\lambda}(\mathbb{G},\mathbb{P})$} is the unique solution to the $\mathbb{G}$-BSDE
\begin{align} \label{eq:BSDEGTheorem}
    Y_{\tau \wedge t}=  - \mathcal{L}_{\tau \wedge T} {+} \int_{\tau \wedge t}^{\tau \wedge T}  \sum_{i=1}^d  f^i(u,Z_u^i) du  - \sum_{i=1}^d\int_{\tau \wedge t}^{\tau \wedge T} Z^{i}_u dW_u^{ f,i} - \sum_{j \in \lbrace B,C  \rbrace} \int_{\tau \wedge t}^{\tau \wedge T} U_u^jdM_u^j
\end{align}
with $f^i:(\Omega \times [0,T] \times \mathbb{R},\mathscr{P}\otimes \mathcal{B}(\mathbb{R})) \to (\mathbb{R},\mathcal{B}(\mathbb{R}))$ defined by
\begin{align} \label{eq:ChoiceOfFDefaultsTheorem}
f^i(\omega, u,z):= -  z_u\frac{\mu^i(u,S_u(\omega))-r_u^i}{\sigma^{i}(u,S_u(\omega))}.
\end{align}
Here, $\mathscr{P}$ denotes the predictable $\sigma$-algebra on $\Omega \times [0,T].$
\end{theorem}

Before proving this theorem we state an auxiliary result which guarantees existence and uniqueness of a solution to the $\mathbb{G}$-BSDE in \eqref{eq:BSDEGTheorem} and we link this solution to the solution of an $\mathbb{F}$-FBSDE.
\begin{proposition} \label{prop:ExistencUniquenessBSDE}
    Let $X^{\tau}=(X^{\tau}_t)_{t \in [0, \tau \wedge T]}$ and $\mathcal{L}=(\mathcal{L}_t)_{t \in [0,\tau \wedge T]}$ be as in \eqref{eq:SemimartingaleDecomposotionX} and \eqref{eq:PaymentStreamL}, respectively. Then, the $\mathbb{G}$-BSDE in \eqref{eq:BSDEGTheorem} has a unique solution $(Y,Z,U^C, U^B) \in \mathbb{S}^2(\mathbb{G}, \mathbb{P}) \times \mathcal{H}^{2,d}(\mathbb{G},\mathbb{P}) \times \mathbb{H}^{2,2}_{\lambda}(\mathbb{G},\mathbb{P})$. Moreover, let $(\overline{Y},\overline{Z}) \in  \mathbb{S}^2(\mathbb{F}, \mathbb{P}) \times \mathcal{H}^{2,d}(\mathbb{F}, \mathbb{P})$ be the unique solution of the $\mathbb{F}$-BSDE
    \begin{align}
d\overline{Y}_t&=-\sum_{i=1}^d  \overline{f}^i(t, \overline{Y}_t,\overline{Z}_t^i)dt {+} \sum_{i=1}^d \overline{Z}_t^i d W_t^{f,i},  \label{eq:FBSDE1} \\
\overline{Y}_T&= \overline{\mathcal{L}}_T, \label{eq:FBSDE2}
\end{align}
with $\overline{f}^i:(\Omega \times [0,T] \times \mathbb{R} \times \mathbb{R}, \mathscr{P} \otimes \mathcal{B}(\mathbb{R}^2))  \to (\mathbb{R},\mathcal{B}(\mathbb{R}))$, $\phi^C=(\phi^C_t)_{t \in [0,T]} \in \mathbb{L}^2(\mathbb{F}, \mathbb{P}), \phi^B=(\phi^B_t)_{t \in [0,T]}\in \mathbb{L}^2(\mathbb{F}, \mathbb{P})$ and $\overline{\mathcal{L}}_T \in L^2(\mathcal{F}_T, \mathbb{P})$ defined by
\allowdisplaybreaks
\begin{align} %\label{eq:DriverWithoutDefaults}
 \overline{f}^i(\omega,u,y,z)&:= f^i(\omega, u, z) + (\phi^C_u-y)\lambda_u^C+ (\phi^B_u-y)\lambda_u^B, \nonumber  \\
\overline{\mathcal{L}}_T&:=  {-}\int_{0}^{ T} \frac{1}{B_u^f} \left(dA_u + dC_u + \bar{r}_u^c C_u du\right) + \frac{C_T}{B_T^f} , \nonumber \\
\phi_t^C&:={-}{\int_0^t \frac{1}{B_u^f}\left(dA_u +dC_u + \bar{r}_u^c C_u du   \right) + \frac{(R^C \mathcal{Y}^+-\mathcal{Y}^-)}{B_{t}^f},} \nonumber \\
\phi_t^B&:={-} { \int_0^t \frac{1}{B_u^f}\left(dA_u +dC_u + \bar{r}_u^c C_u du   \right) + \frac{(\mathcal{Y}^+-R^B \mathcal{Y}^-)}{B_{t}^f},}\nonumber 
\end{align}
and $f^i$ given in \eqref{eq:ChoiceOfFDefaultsTheorem}.
On the one hand, the solution $(Y,Z,U^C, U^B) \in \mathbb{S}^2(\mathbb{G}, \mathbb{P}) \times \mathcal{H}^{2,d}(\mathbb{G},\mathbb{P}) \times \mathcal{H}^{2,2}_{\lambda}(\mathbb{G},\mathbb{P})$ of the $\mathbb{G}$-BSDE \eqref{eq:BSDEGTheorem} can be represented as
\begin{align}
    Y_t&= \overline{Y}_t \textbf{1}_{\lbrace t <\tau \wedge T \rbrace} 
+ \left( \phi^C_{\tau^C} \textbf{1}_{\lbrace \tau^C < \tau^B \wedge T \rbrace} + \phi^B_{\tau^B} \textbf{1}_{\lbrace \tau^B < \tau^C \wedge T \rbrace} + \overline{\mathcal{L}}_T \textbf{1}_{\lbrace \tau >T \rbrace}\right) \textbf{1}_{\lbrace t= \tau \wedge T \rbrace}, \nonumber \\
Z_t^i&=\overline{Z}_t^i \textbf{1}_{\lbrace t <\tau \wedge T \rbrace}, \quad i=1,...,d,\label{eq:RepresentationBSDEs1} \\
U^j_t&= \left(\phi^j_t-\overline{Y}_t\right) \textbf{1}_{\lbrace t <\tau \wedge T \rbrace}, \quad j \in \lbrace B, C \rbrace.\nonumber 
\end{align}
On the other hand, the solution of the $\mathbb{F}$-BSDE \eqref{eq:FBSDE1}-\eqref{eq:FBSDE2} is given by
\begin{align}\label{eq:RepresentationBSDEs2}
    \overline{Y}_t&= Y_{t \wedge {\tau}-}\nonumber, \\
    \overline{Z}_t^i&=Z_t^i \textbf{1}_{ \lbrace t< \tau \wedge T\rbrace}, \quad i=1,...,d.
\end{align}
\end{proposition}
\begin{proof}
First note that the $\mathbb{F}$-BSDE \eqref{eq:FBSDE1}-\eqref{eq:FBSDE2} has a unique solution $(\overline{Y},\overline{Z}) \in \mathbb{S}^2(\mathbb{F}, \mathbb{P}) \times \mathcal{H}^{2,d}(\mathbb{F}, \mathbb{P}).$ This follows directly by Theorem 3.1 in \cite{delong} as ${\overline{\mathcal{L}}_T \in L^2(\mathcal{F}_T, \mathbb{P})}$ and for $(\omega,t,y,z), (\omega,t,\bar{y},\bar{z}) \in \Omega \times [0,T] \times \mathbb{R}^2$  the driver satisfies 
\allowdisplaybreaks
\begin{align*}
   \left \vert \sum_{i=1}^d \overline{f}^i(\omega, t,y,z) - \sum_{i=1}^d \overline{f}^i(\omega, t,\bar{y}, \bar{z})  \right \vert^2 
   &= \left \vert\sum_{i=1}^d  \left(( f^i (t,z)-f^i(t,\bar{z}))+  \lambda_t^C(\bar{y}-y)+ \lambda_t^B(\bar{y}-y)\right) \right \vert^2   \nonumber \\
    &\leq 2d \sum_{i=1}^d \left( \left \vert f^i (t,z)-f^i(t,\bar{z}) \right \vert^2  + \left \vert  \lambda_t^C + \lambda_t^B \right \vert^2 \left \vert y- \bar{y}\right \vert^2 \right)  \nonumber \\
    &\leq 2 d  \sum_{i=1}^d \left(\max_{i=1,...,d}K_i^2\vert z- \bar{z} \vert^2 +  2\sup_{t \in [0,T] }\left \vert \lambda_t^C+ \lambda_t^B \right\vert^2 \left \vert y- \bar{y}\right \vert^2 \right),
\end{align*}
and 
\begin{align*}
    \mathbb{E}_{\mathbb{P}} \left[ \int_0^T \left \vert \sum_{i=1}^d \overline{f}^i(t,0,0)\right \vert^2 dt\right]&=\mathbb{E}_{\mathbb{P}} \left[ \int_0^T \left \vert \sum_{i=1}^d \left( f^i(t,0) + \phi_t^C \lambda_t^C + \phi_t^B \lambda_t^B\right) \right \vert^2 dt\right]  \nonumber \\
    & \leq d^2 \max_{j \in \lbrace B, C \rbrace} \sup_{t \in [0,T]} \left \vert \lambda_t^j \right \vert^2 \mathbb{E}_{\mathbb{P}} \left[ \int_0^T \left \vert \phi_t^C + \phi_t^B \right \vert^2 dt\right]\nonumber \\
    & < \infty, 
\end{align*}
where we use \eqref{eq:BoundsExistenceSolutionDefaults},  $(\mathcal{L}_{t})_{t \in [0,\tau \wedge T]} \in \mathbb{L}^2(\mathbb{G}, \mathbb{P})$ and that the intensity processes $\lambda^B, \lambda^C$ are bounded by assumption.\\
Next, we rewrite the $\mathbb{G}$-BSDE in \eqref{eq:BSDEGTheorem} as 
\begin{align}
dY_t&= - \sum_{i=1}^d  f^i(t,Z_t)dt {+} \sum_{i=1}^d Z_t^i d W_t^{f,i} {+} \sum_{j \in \lbrace B,C  \rbrace} U_t^jdM_t^j, \quad t \in [0, \tau \wedge T], \label{eq:TerminalConditionG_0} \\ 
Y_{\tau \wedge T}&= - \int_{0}^{\tau \wedge T} \frac{1}{B_u^f} dA_u^{R,C} +\frac{C_{T}}{B_{T}^f} \textbf{1}_{\lbrace \tau >T \rbrace  } - \frac{R_{\tau}}{B_{\tau}^f} \textbf{1}_{\lbrace \tau \leq T \rbrace}. \label{eq:TerminalConditionG}
\end{align}
To prove the equivalence of the $\mathbb{F}$-BSDE in \eqref{eq:FBSDE1}-\eqref{eq:FBSDE2} up to $\tau-$ to the $\mathbb{G}$-BSDE in \eqref{eq:TerminalConditionG_0}-\eqref{eq:TerminalConditionG}, we apply Theorem 4.3 in \cite{crepey_song_2015}. To do so, we briefly verify that the necessary conditions of this result are satisfied. By Hypothesis \ref{hp:H} the immersion property holds, which guarantees that Condition (A) in \cite{crepey_song_2015} holds. Condition (J) is always satisfied and condition (B) holds due to \eqref{eq:PreDefault}. Finally, by rewriting the terminal condition in \eqref{eq:TerminalConditionG} we have
\begin{align}
&- \int_{0}^{\tau \wedge T} \frac{1}{B_u^f} dA_u^{R,C} +\frac{C_{T}}{B_{T}^f} \textbf{1}_{\lbrace \tau >T \rbrace  } - \frac{R_{\tau}}{B_{\tau}^f} \textbf{1}_{\lbrace \tau \leq T \rbrace} \nonumber \\
&= - \int_{0}^{\tau \wedge T} \frac{1}{B_u^f} dA_u^{R,C} +\frac{C_{T}}{B_{T}^f} \textbf{1}_{\lbrace \tau >T \rbrace  } -  \textbf{1}_{\lbrace{\tau^C< \tau^B \wedge T \rbrace}}\frac{(R^C \mathcal{Y}^+-\mathcal{Y}^-)}{B_{\tau^C}^f}+\textbf{1}_{\lbrace{\tau^B< \tau^C \wedge T \rbrace }}\frac{(\mathcal{Y}^+-R^B \mathcal{Y}^-)}{B_{\tau^B}^f} \nonumber \\
&=\left( - \int_{0}^{\tau \wedge T} \frac{1}{B_u^f} dA_u^{R,C} + \frac{C_T}{B_T^f} \right)\textbf{1}_{\lbrace \tau >T \rbrace  }-\left(  \int_{0}^{\tau \wedge T} \frac{1}{B_u^f} dA_u^{R,C} +  \frac{(R^C \mathcal{Y}^+-\mathcal{Y}^-)}{B_{\tau^C}^f}\right)\textbf{1}_{\lbrace{\tau^C< \tau^B \wedge T \rbrace}} \nonumber \\
& \quad - \left(  \int_{0}^{\tau \wedge T} \frac{1}{B_u^f} dA_u^{R,C} + \frac{(\mathcal{Y}^+-R^B \mathcal{Y}^-)}{B_{\tau^B}^f} \right)\textbf{1}_{\lbrace{\tau^B< \tau^C \wedge T \rbrace }} \nonumber \\
&=\left( - \int_{0}^{ T} \frac{1}{B_u^f} dA_u^{C} + \frac{C_T}{B_T^f} \right)\textbf{1}_{\lbrace \tau >T \rbrace  }-\left(  \int_{0}^{\tau^C} \frac{1}{B_u^f} dA_u^{C} +  \frac{(R^C \mathcal{Y}^+-\mathcal{Y}^-)}{B_{\tau^C}^f}\right)\textbf{1}_{\lbrace{\tau^C< \tau^B \wedge T \rbrace}} \nonumber \\
& \quad - \left(  \int_{0}^{\tau^B} \frac{1}{B_u^f} dA_u^{C} + \frac{(\mathcal{Y}^+-R^B \mathcal{Y}^-)}{B_{\tau^B}^f} \right)\textbf{1}_{\lbrace{\tau^B< \tau^C \wedge T \rbrace }} \nonumber \\
&=\overline{\mathcal{L}}_T\textbf{1}_{\lbrace \tau >T \rbrace  }- \phi^C_{\tau^C}\textbf{1}_{\lbrace{\tau^C< \tau^B \wedge T \rbrace}} + \phi^B
_{\tau^B}\textbf{1}_{\lbrace{\tau^B< \tau^C \wedge T \rbrace }}. \nonumber 
\end{align}
Thus, the representations in \eqref{eq:RepresentationBSDEs1} and \eqref{eq:RepresentationBSDEs2} follow immediately.
\end{proof}
We now prove Theorem \ref{theorem:LinkBSDE}.
\begin{proof}
By Proposition \ref{prop:Existence} there exists a unique $\mathbb{G}$-F\"ollmer-Schweizer decomposition of $\mathcal{L}_{\tau \wedge T}$. 
We start by proving that the decomposition of $\mathcal{L}_{\tau \wedge T} $ in \eqref{eq:FoellmerSchweizerBSDEtheorem} and \eqref{eq:BSDEGTheorem} satisfies the definition of a $\mathbb{G}$-F\"ollmer-Schweizer decomposition. First, we verify that $\xi^{\mathcal{L}}$ defined in \eqref{eq:SepcificationFSXi} is in the space $\Theta^{\mathbb{G}, \tau}_X$. By the regularity of $Z$ as a control process, it follows that $\xi^{\mathcal{L}}$ is $\mathbb{G}$-predictable. 
Thus, it remains to verify that \eqref{eq:integrabilityTau} holds. For fixed $i=1,...,d$ it holds
\allowdisplaybreaks
\begin{align}
    &\mathbb{E}_{\mathbb{P}} \left[\sum_{i=1}^d \int_0^{\tau \wedge T} \left \vert \xi_u^{\mathcal{L},i}\right \vert^2 d \langle M^{X^{i}}\rangle_u +\left( \sum_{i=1}^d \int_0^{\tau \wedge T} \left \vert \xi_u^{\mathcal{L},i}  dA_{u}^{X^{i}}\right \vert \right)^2\right] \nonumber \\
    &=\mathbb{E}_{\mathbb{P}} \left[\sum_{i=1}^d \int_0^{\tau \wedge T} \left \vert \frac{Z_u^i B_u^f}{S_u^{i} {\sigma}^i(u,S_u)}\right \vert^2 \left \vert \frac{S_u^{i}}{B_u^f}\right\vert^2 \left\vert{\sigma}^{i}(u,S_u)\right\vert^2 du \right] \nonumber \\
    & \quad +\mathbb{E}_{\mathbb{P}}\left[\left(\sum_{i=1}^d  \int_0^{\tau \wedge T} \left \vert \frac{Z_u^i B_u^f}{S_u^{i} {\sigma}^i(u,S_u)}  \frac{S_u^{i}}{B_u^f}  \left(\mu^i(u,S_u) -r_u^i \right)\right \vert du\right)^2 \right] \nonumber \\
    &=\mathbb{E}_{\mathbb{P}} \left[ \sum_{i=1}^d\int_0^{\tau \wedge T} \left(Z_u^i\right)^2 du +\left(\sum_{i=1}^d \int_0^{\tau \wedge T} \left\vert \frac{\mu^i(u,S_u) -r_u^i}{\sigma^i(u,S_u)}Z_u^i \right \vert du\right)^2 \right] \nonumber \\ 
    & \leq \|Z \|^2_{\mathcal{H}^{2,d}} + d \max_{i=1,...,d} K_i^2 \; \mathbb{E}_{\mathbb{P}} \left[ \sum_{i=1}^d \int_0^{\tau \wedge T} \left\vert Z_u^i \right \vert^2 du \right]\nonumber\\
    & \leq \left(1+ d \max_{i=1,...,d} K_i^2\right)\|Z \|^2_{\mathcal{H}^{2,d}}, \nonumber 
\end{align}
where we use \eqref{eq:BoundsExistenceSolutionDefaults} and Jensen's inequality. 

Next, we show that $\mathcal{H}^{\mathcal{L}}$ is strongly orthogonal to the martingale part of $X^{\tau}$. Note that by the regularity of $U^C$ and $U^B$ the processes $(\int_0^t U_u^C dM_u^C)_{t \in [0,\tau \wedge T]}$ and $(\int_0^t U_u^B dM_u^B)_{t \in [0,\tau \wedge T]}$  are square-integrable $\mathbb{G}$-martingales and strongly orthogonal to each other, see e.g. similar arguments as in Section 5.2.2 in \cite{BieRut15}.
This implies that $({H}^{\mathcal{L}}_t)_{t \in [0,\tau \wedge T]}= ( \int_0^t U_u^C dM_u^C + \int_0^t U_u^B dM_u^B)_{t \in [0,\tau \wedge T]}$ is strongly orthogonal to the martingale component of $X^{\tau}$, as the martingale component of $X^{\tau}$ is driven by $W^{\tau,f}$ and $X^{\tau}$ is sufficiently integrable.
\allowdisplaybreaks
Thus, it only remains to verify the decomposition. By rearranging the $\mathbb{G}$-BSDE in \eqref{eq:BSDEGTheorem} we get
\begin{align} 
   -\mathcal{L}_{\tau \wedge T}&=Y_0- \int^{\tau \wedge T}_{0} \sum_{i=1}^d f^i(u,Z_u^i) du + \sum_{i=1}^d \int^{\tau \wedge T}_{0} Z_u^i dW_u^{f,i} + \sum_{j \in \lbrace B, C \rbrace}\int^{\tau \wedge T}_{0} U_u^j dM_u^j \nonumber \\
   &= Y_0 + \int^{\tau \wedge T}_{0} \sum_{i=1}^d   Z_u^i\frac{\mu^i(u,S_u)-r_u^i}{\sigma^{i}(u,S_u)} du + \sum_{i=1}^d \int^{\tau \wedge T}_{0} Z_u^i dW_u^{f,i} + \sum_{j \in \lbrace B, C \rbrace}\int^{\tau \wedge T}_{0} U_u^j dM_u^j \nonumber  \\
   & =Y_0 + \int^{\tau \wedge T}_{0} \sum_{i=1}^d   Z_u^i\frac{B_u^f}{\sigma^{i}(u,S_u)S_u^{i}} \frac{B_u^i}{B_u^f}\left[\left(\mu^i(u,S_u)-r_u^i\right) \frac{S_u^{i}}{B_u^i} du + \sigma^i(u,S_u) \frac{S_u^{i}}{ B_u^i}dW_u^{f,i}\right]du \nonumber \\
   & \quad + \sum_{j \in \lbrace B, C \rbrace}\int^{\tau \wedge T}_{0} U_u^j dM_u^j \nonumber \\
   &=Y_0 + \int^{\tau \wedge T}_{0} \sum_{i=1}^d   Z_u^i\frac{B_u^f}{\sigma^{i}(u,S_u)S_u^{i}} \frac{B_u^i}{B_u^f}d \widetilde{S}^{i}_u  + \sum_{j \in \lbrace B, C \rbrace}\int^{\tau \wedge T}_{0} U_u^j dM_u^j \label{eq:UseDynamicsDiscountedS} \\
   &= Y_0 + \sum_{i=1}^d \int_0^T \xi^{\mathcal{L},i}_u \frac{B_u^i}{B_u^f} d\widetilde{S}_u^{\tau,i} + H_{\tau \wedge T}^{\mathcal{L}} \label{eq:UseDynamicsDiscountedS_1}
\end{align}
where \eqref{eq:UseDynamicsDiscountedS} follows by \eqref{eq:SemimartingaleDecompositionS}
%\begin{equation*} 
%\widetilde{S}^{\tau,i}_t=\widetilde{S}^i_0+ \int_0^{t} \left( \mu^i(u,S_u^{\tau}) - r_u^i \right) \widetilde{S}^{\tau,i}_u du + \int_0^{t} \sigma^i(u,S_u^{\tau}) \widetilde{S}^{\tau,i}_udW_u^{\tau, f,i}, \quad t \in [0, \tau \wedge T]
%\end{equation*}
and $\xi^{\mathcal{L},i}$ and $\mathcal{H}^{\mathcal{L}}$ in \eqref{eq:UseDynamicsDiscountedS_1} are defined as in \eqref{eq:SepcificationFSXi}.
\end{proof}
\begin{lemma} \label{lemma:DescriptionStrategyBSDE}
Let $X^{\tau}=(X^{\tau}_t)_{t \in [0, \tau \wedge T]}$ and $\mathcal{L}=(\mathcal{L}_t)_{t \in [0,\tau \wedge T]}$ be as in \eqref{eq:SemimartingaleDecomposotionX} and \eqref{eq:PaymentStreamL}, respectively.
Then the locally risk-minimizing $L_{\mathbb{G}}^2$-strategy $\varphi=(\xi,\psi,\psi^f)$ for the defaultable, collateralized contract $(A,R,C,\tau)$  is given by
\begin{align*}
\xi_t^i&:=\xi_t^{\mathcal{L},i}= \frac{Z_t^{i}B_t^f}{S_t^{i}\sigma^i(t, S_t)}, \quad \psi_t^i = - \frac{Z_t^i B_t^f}{\sigma^i(t,S_t) B_t^i}, \quad i=1,...,d, \; t \in [0, \tau \wedge T], \\
\psi_t^f& = \widetilde{V}_t^p(\varphi, A^{R,C})=Y_t + \int_0^t \frac{1}{B_u^f}dA_u^{R,C}, \quad  t \in [0, \tau \wedge T],
\end{align*}
where {$(Y,Z,U^C, U^B) \in \mathbb{S}^2(\mathbb{G}, \mathbb{P}) \times \mathcal{H}^{2,d}(\mathbb{G},\mathbb{P}) \times \mathcal{H}^{2,2}_{\lambda}(\mathbb{G},\mathbb{P})$} is the unique solution to the $\mathbb{G}$-BSDE in \eqref{eq:BSDEGTheorem}.
\end{lemma}
\begin{proof}
 By Theorem \ref{theorem:LinkBSDE} $\mathcal{L}_{\tau \wedge T}$ admits a unique $\mathbb{G}$-F\"ollmer-Schweizer decomposition which is specified by \eqref{eq:FoellmerSchweizerBSDEtheorem}-\eqref{eq:ChoiceOfFDefaultsTheorem}. Furthermore, by Proposition \ref{prop:EquivalenceLRMFS} the associated locally risk-minimizing $L_{\mathbb{G}}^2$-strategy $\varphi=(\xi,\psi,\psi^f)$ is given by \eqref{eq:SpecificationLRMS1}-\eqref{eq:SpecificationLRMS2}. Thus, by \eqref{eq:ValuePortfolioNoDefault} it holds for $t \in [0, \tau \wedge T] $
  \begin{align}
        \psi^f_t&= Y_0 + \sum_{i=1}^d  \int_{0}^{t}  {\xi}_u^{\mathcal{L},i} \frac{B_u^i}{B_u^f} d \widetilde{S}_{u}^{\tau,i} + {H}_{t}^{\mathcal{L}} + \int_0^{t} \frac{1}{B_u^f} dA_u^{R,C} \nonumber \\
        &=Y_0  + \sum_{i=1}^d  \int_{0}^{t} \frac{Z_u^{i}B_u^f}{S_u^{i}\sigma^i(u, S_u)} \frac{B_u^i}{B_u^f} d \widetilde{S}_{u}^{\tau,i} + \sum_{j \in \lbrace{B,C \rbrace}} \int_0^t U_u^j dM_u^{j} + \int_0^{t} \frac{1}{B_u^f} dA_u^{R,C} \nonumber \\
        &= Y_0 -\int_0^t \sum_{i=1}^d f^i(u,Z_u^i)du + \sum_{i=1}^d \int_0^t Z_u^i dW_u^{\tau,f,i}
        +\sum_{j \in \lbrace{B,C \rbrace}}\int_0^t  U_u^j dM_u^{j} + \int_0^{t} \frac{1}{B_u^f} dA_u^{R,C}, \label{eq:RepresentationPsiBSDE}  
    \end{align}
    by using for the last equality the same arguments as in \eqref{eq:UseDynamicsDiscountedS}. For $t \in [0,\tau \wedge T]$ we then derive by \eqref{eq:BSDEGTheorem} that 
    $$
    \psi_t^f=Y_t + \int_0^t \frac{1}{B_u^f}dA_u^{R,C}.
    $$
\end{proof}

\subsection{Value adjustments} \label{sec:ValueAdjustments}
In this section we derive a decomposition of the price in terms of value adjustments.

\begin{assumption} \label{assump:cleanValue}
We assume that the clean value process $\mathcal{V} \in \mathbb{S}^2(\mathbb{F}, \mathbb{P})$, i.e. 
\begin{align} \nonumber
\mathbb{E}_{\mathbb{P}}\left[ \sup_{t \in [0,T]} \vert \mathcal{V}_t \vert^2\right]< \infty.
\end{align} 
\end{assumption}
\begin{remark}
\begin{enumerate}
\item Note that Assumption \ref{assump:cleanValue} is a quite weak condition, as it can be shown that the clean value process $\mathcal{V}$ is a solution of a $\mathbb{F}$-BSDE and thus $\mathcal{V} \in \mathbb{S}^2(\mathbb{F}, \mathbb{P})$, see e.g. Theorem 3.5 in \cite{Biagini_Gnoatto_Oliva_2021}. 
\item Assumption \ref{assump:cleanValue} implies that 
\begin{align} \label{eq:IntegrabilityCollateral}
    \mathbb{E}_{\mathbb{P}}\left[ \sup_{t \in [0,T]} \vert C_t \vert^2 \right]< \infty.
\end{align}
This follows as 
\begin{align}
    \mathbb{E}_{\mathbb{P}}\left[ \sup_{t \in [0,T]} \vert C_t \vert^2 \right]&=\mathbb{E}_{\mathbb{P}}\left[ \sup_{t \in [0,T]} \vert g(\mathcal{V}_t) \vert^2 \right] \leq \mathbb{E}_{\mathbb{P}}\left[ \sup_{t \in [0,T]} \Big \vert K_g  \vert \mathcal{V}_t \vert + \vert f(0) \vert \Big \vert^2 \right] \nonumber \\
    &\leq 2 (K_g)^2 \mathbb{E}_{\mathbb{P}}\left[ \sup_{t \in [0,T]}  \vert \mathcal{V}_t \vert^2 \right] + 2  \vert f(0) \vert^2 < \infty, \nonumber 
\end{align}
where $K_g$ denotes the Lipschitz constant of $g$.   
\end{enumerate}
\end{remark}
Moreover, we introduce the processes $\widehat{V}(\varphi,A^{R,C})=(\widehat{V}_t(\varphi,A^{R,C}))_{t \in [0, \tau \wedge T]}$ and \linebreak $\widehat{V}^p(\varphi,A^{R,C})=(\widehat{V}_t^p(\varphi,A^{R,C}))_{t \in [0, \tau \wedge T]}$ by 
\begin{align} \nonumber 
    \widehat{V}_t(\varphi,A^{R,C}):= \frac{{V}_t(\varphi,A^{R,C})}{B_t^r} \quad \text{and} \quad \widehat{V}_t^p(\varphi,A^{R,C}):= \frac{{V}^p_t(\varphi,A^{R,C})}{B_t^r},
\end{align}
respectively.
\begin{definition} \label{def:ValuationAdjustments}
We define the following valuation adjustment on $\lbrace{ t< \tau \rbrace}$
\begin{align}
    \textnormal{CVA}_t&:= {\mathbb{E}_{{\mathbb{Q}}}\left[ \textbf{1}_{\lbrace \tau \leq T \rbrace} \textbf{1}_{\lbrace{\tau^C< \tau^B\rbrace}}\frac{B_t^r}{B_{\tau}^r}(1-R^C) (Q_{\tau}-C_{\tau-})^+\Big \vert \mathcal{G}_t\right] } \nonumber \\
    &= {\frac{B_t^r}{ \mathscr{Z}_t}\mathbb{E}_{\mathbb{P}}\left[ \textbf{1}_{\lbrace \tau \leq T \rbrace} \textbf{1}_{\lbrace{\tau^C< \tau^B\rbrace}}\frac{\mathscr{Z}_{\tau}}{B_{\tau}^r}(1-R^C) (Q_{\tau}-C_{\tau-})^+\Big \vert \mathcal{G}_t\right] },  \nonumber \\
    \textnormal{DVA}_t&:=\mathbb{E}_{{\mathbb{Q}}}\left[ \textbf{1}_{\lbrace \tau \leq T \rbrace}\textbf{1}_{\lbrace{\tau^B< \tau^C\rbrace}}\frac{B_t^r}{B_{\tau}^r} (1-R^B)(Q_{\tau}-C_{\tau-})^-\Big \vert \mathcal{G}_t\right]\nonumber \\
    &=\frac{B_t^r}{ \mathscr{Z}_t}\mathbb{E}_{\mathbb{P}}\left[ \textbf{1}_{\lbrace \tau \leq T \rbrace}\textbf{1}_{\lbrace{\tau^B< \tau^C\rbrace}}\frac{\mathscr{Z}_{\tau}}{B_{\tau}^r} (1-R^B)(Q_{\tau}-C_{\tau-})^-\Big \vert \mathcal{G}_t\right], \nonumber \\
    \textnormal{CollVA}_t&:=\mathbb{E}_{{\mathbb{Q}}}\left[ \int_{t}^{\tau \wedge T} C_u\frac{B_t^r}{B_u^r}\left(\bar{r}_u^c-r_u^f \right) du\Big \vert \mathcal{G}_t\right]=\frac{B_t^r}{ \mathscr{Z}_t}\mathbb{E}_{\mathbb{P}}\left[ \int_{t}^{\tau \wedge T} \mathscr{Z}_u\frac{C_u}{B_u^r}\left(\bar{r}_u^c-r_u^f \right) du\Big \vert \mathcal{G}_t\right], \nonumber \\
    \textnormal{FVA}_t&:=\mathbb{E}_{{\mathbb{Q}}}\left[ \int_{t}^{\tau \wedge T} V_u(\varphi, A^{R,C})\frac{B_t^r}{B_u^r}\left({r}_u^f-r_u^r \right) du\Big \vert \mathcal{G}_t\right]=\frac{B_t^r}{ \mathscr{Z}_t}\mathbb{E}_{\mathbb{P}}\left[ \int_{t}^{\tau \wedge T} \mathscr{Z}_u\frac{V_u(\varphi, A^{R,C})}{B_u^r}\left({r}_u^f-r_u^r \right) du\Big \vert \mathcal{G}_t\right]. \nonumber 
\end{align}
\end{definition}

\begin{theorem} \label{theorem:RepresentatinXVA}
Let $X^{\tau}=(X^{\tau}_t)_{t \in [0, \tau \wedge T]}$ and $\mathcal{L}=(\mathcal{L}_t)_{t \in [0,\tau \wedge T]}$ be as in \eqref{eq:SemimartingaleDecomposotionX} and \eqref{eq:PaymentStreamL}, respectively, and Assumption \ref{assump:cleanValue} hold. Furthermore, we assume that 
{\begin{align}\label{eq:AdditionalAssumption1}
    \mathbb{E}_{\mathbb{P}}\left[\int_{0}^{\tau \wedge T} \left \vert \int_0^u \frac{1}{B_s^f} d A^{R,C}_s \right\vert^2 du\right] < \infty.
\end{align}}Denote by $\varphi=(\xi,\psi,\psi^f)$ the locally risk-minimizing $L_{\mathbb{G}}^2$-strategy for the defaultable, collateralized contract $(A,R,C,\tau)$. Then it holds on $\lbrace{t < \tau \rbrace}$
\begin{align}
    V_t(\varphi, A^{R,C})= \mathcal{V}_t+ \textnormal{CVA}_t-\textnormal{DVA}_t+ \textnormal{CollVA}_t-\textnormal{FVA}_t. 
\end{align}
\end{theorem}
\begin{proof}
    By \eqref{eq:RepresentationPsiBSDE} we have on $ \lbrace{ t < \tau \rbrace}$
    \begin{align}
        \widetilde{V}_t^p(\varphi, A^{R,C})
        &=Y_0 + \sum_{i=1}^d  \int_{0}^t Z_u^{i} \frac{\mu^i (u,S_u)-r_u^i}{\sigma^i(u, S_u)}du +\sum_{i=1}^d   \int_{0}^t Z_u^i dW_u^{f,i} + {H}_{t }^{\mathcal{L}} + \int_0^t \frac{1}{B_u^f} dA_u^{R,C}. \nonumber
    \end{align}
Moreover, by It\^{o}'s formula it holds 
\begin{align} \label{eq:DynamicsVp}
    d{V}_t^p(\varphi, A^{R,C})&= 
     B_t^f d\widetilde{V}_t^p(\varphi, A^{R,C}) + r_t^f {V}_t^p(\varphi, A^{R,C})dt
    \nonumber \\
    &=B_t^f \left( \sum_{i=1}^d Z_u^i\frac{\mu^i(t,S_t)-r_t^i}{\sigma^{i}(t,S_t)} dt +  \sum_{i=1}^d    Z_t^i dW_t^{f,i}+ dH_t^{\mathcal{L}}+ \frac{1}{B_t^f}dA_t^{R,C}\right) \nonumber \\
    & \quad + \left( r_t^f-r_t \right){V}_t^p(\varphi, A^{R,C}) dt + r_t {V}_t^p(\varphi, A^{R,C})dt,
\end{align}
and thus if follows
\begin{align}
d \widehat{V}_t^p(\varphi, A^{R,C})&=\frac{1}{B_t^r}d {V}_t^p(\varphi, A^{R,C})-\frac{1}{B_t^r} r_t {V}_t^p(\varphi, A^{R,C})dt \nonumber \\
&=\frac{B_t^f}{B_t^r}\sum_{i=1}^d Z_u^i \frac{\mu^i(t,S_t)-r_t^i}{\sigma^{i}(t,S_t)} dt + \frac{B_t^f}{B_t^r}\sum_{i=1}^d    Z_t^i dW_t^{ f,i}+ \frac{B_t^f}{B_t^r}dH_t^{\mathcal{L}}+ \frac{1}{B_t^r}dA_t^{R,C} \nonumber \\
& \quad +\left( r_t^f-r_t \right)\frac{{V}_t^p(\varphi, A^{R,C})}{B_t^r} dt. \nonumber
\end{align}
Again by It\^{o}'s formula we get
    \begin{align}
        &d\left(\widehat{V}_t^p(\varphi, A^{R,C}) \mathscr{Z}_t \right)\nonumber \\
        &=  \mathscr{Z}_t\left(\frac{B_t^f}{B_t^r}\sum_{i=1}^d Z_u^i\frac{\mu^i(t,S_t)-r_t^i}{\sigma^{i}(t,S_t)} dt +  \frac{B_t^f}{B_t^r}\sum_{i=1}^d    Z_t^i dW_t^{ f,i}+\frac{B_t^f}{B_t^r} dH_t^{\mathcal{L}}+ \frac{1}{B_t^r}dA_t^{R,C}+\left( r_t^f-r_t \right)\frac{{V}_t^p(\varphi, A^{R,C})}{B_t^r} dt\right) \nonumber \\
        & \quad + \widehat{V}_t^p(\varphi, A^{R,C})\mathscr{Z}_t\left( \sum_{i=1}^d -\frac{\mu^i(t,S_t)-r_t^i}{\sigma^{i}(t,S_t)} dW_t^{f,i}\right)
       {+ d\left [ \widehat{V}^p_{\cdot}(\varphi, A^{R,C}), \mathscr{Z}_{\cdot} \right ]_t},\label{eq:QuadraticVariationProduct} 
    \end{align}
    where $\mathscr{Z}$ is defined in \eqref{eq:StochasticExponential}.
    {On $\lbrace t < \tau \rbrace$ we have
    \allowdisplaybreaks
    \begin{align}
         &d\left [ \widehat{V}^p(\varphi, A^{R,C}), \mathscr{Z} \right ]_t \nonumber \\
         &= d \left \langle h_0 + \sum_{i=1}^d  \int_{0}^{\cdot} \frac{B_u^f}{B_u^r}Z_u^{i} \frac{\mu^i (u,S_u)-r_u^i}{\sigma^i(u, S_u)}du +\sum_{i=1}^d   \int_{0}^{\cdot}\frac{B_u^f}{B_u^r}  Z_u^i dW_u^{f,i} + \int_0^{\cdot} \frac{B_u^f}{B_u^r} d{H}_{u}^{\mathcal{L}} + \int_0^{\cdot} \frac{1}{B_u^r} dA_u^{R,C}, \mathscr{Z}_{\cdot} \right \rangle_t \nonumber \\
         & \quad + d  \left \langle \int_0^{\cdot} \left( r_u^f-r_u \right) \widehat{V}_u^p(\varphi, A^{R,C}) du, \mathscr{Z}_{\cdot} \right \rangle_t \nonumber \\
         &=  - \mathscr{Z}_t  \frac{B_t^f}{B_t^r}\sum_{i=1}^d Z_t^i\frac{\mu^i(t,S_t)-r_t^i}{\sigma^{i}(t,S_t)} dt + d \left \langle  \int_0^{\cdot} \frac{1}{B_u^r} dA_u^{R,C}, \mathscr{Z}_{\cdot} \right \rangle_t, \nonumber \\
         &=  - \mathscr{Z}_t \frac{B_t^f}{B_t^r} \sum_{i=1}^d Z_t^i\frac{\mu^i(t,S_t)-r_t^i}{\sigma^{i}(t,S_t)} dt + d \left \langle  \int_0^{\cdot} \frac{1}{B_u^r} dA_u + \int_0^{\cdot} \frac{1}{B_u^r} dC_u -\int_0^{\cdot} \frac{1}{B_u^r}\bar{r}_u^c C_u du   , \mathscr{Z}_{\cdot} \right \rangle_t \nonumber\\
        %  &=  - \mathscr{Z}_t \frac{B_t^f}{B_t^r} \sum_{i=1}^d Z_t^i\frac{\mu^i(t,S_t)-r_t^i}{\sigma^{i}(t,S_t)} dt + d \left \langle  \int_0^{\cdot} \frac{1}{B_u^r} dC_u  , \mathscr{Z}_{\cdot} \right \rangle_t \nonumber \\
          &=  - \mathscr{Z}_t \frac{B_t^f}{B_t^r}\sum_{i=1}^d Z_t^i\frac{\mu^i(t,S_t)-r_t^i}{\sigma^{i}(t,S_t)} dt + \frac{1}{B_t^r} d \langle C, \mathscr{Z}\rangle_t \label{eq:Ito1.0},
    \end{align}
  by using in \eqref{eq:Ito1.0} that $C_t (B_t^r)^{-1}+ \int_0^t (B_u^r)^{-1} C_u r_u du = \int_0^t (B_u^r)^{-1} dC_u$.} By combining \eqref{eq:QuadraticVariationProduct} and \eqref{eq:Ito1.0} and using that $\varphi$ is $0$-achieving, i.e. $\widehat{V}_{\tau \wedge T}^p(\varphi, A^{R,C})= \frac{C_T}{B_T^r} \textbf{1}_{\lbrace \tau >T \rbrace} - \frac{R_{\tau}}{B_{\tau}^r}\textbf{1}_{\lbrace \tau \leq T \rbrace}$, it follows
    \begin{align}
        & \left(\frac{C_T}{B_T^r} \textbf{1}_{\lbrace \tau >T \rbrace}- \frac{R_{\tau}}{B_{\tau}^r}\textbf{1}_{\lbrace \tau \leq T \rbrace}\right) \mathscr{Z}_{\tau \wedge T}- \widehat{V}_t^p(\varphi, A^{R,C})\mathscr{Z}_t \nonumber \\
        &=\int_t^{\tau \wedge T} \mathscr{Z}_u\left(\frac{B_u^f}{B_u^r} \sum_{i=1}^d    Z_u^i dW_u^{f,i}+ \frac{B_u^f}{B_u^r}dH_u^{\mathcal{L}}+ \frac{1}{B_u^r}dA_u^{R,C}+ \left( r_u^f-r_u\right)\frac{{V}_u^p(\varphi, A^{R,C})}{B_u^r}du\right)  \nonumber \\
        & \quad +  \int_t^{\tau \wedge T} \widehat{V}_u^p(\varphi, A^{R,C})\mathscr{Z}_u\left( \sum_{i=1}^d -\frac{\mu^i(u,S_u)-r_u^i}{\sigma^{i}(u,S_u)} dW_u^{f,i}\right) + \int_t^{\tau \wedge T} \frac{1}{B_u^r}d \langle C, \mathscr{Z}\rangle_u. \nonumber 
    \end{align}
    By taking the conditional expectation under $\mathbb{P}$ w.r.t. $\mathcal{G}_t$ for $t < \tau$ we have
    \begin{align}
   \widehat{V}_t^p(\varphi, A^{R,C})&= \frac{1}{\mathscr{Z}_t} \mathbb{E}_{\mathbb{P}}\left[ \left(\frac{C_T}{B_T^r} \textbf{1}_{\lbrace \tau >T \rbrace}- \frac{R_{\tau}}{B_{\tau}^r}\textbf{1}_{\lbrace \tau \leq T \rbrace}\right) \mathscr{Z}_{\tau \wedge T} \Big \vert \mathcal{G}_t\right] \nonumber \\
   & \quad - \frac{1}{\mathscr{Z}_t} {\mathbb{E}_{\mathbb{P}}\left[  \int_t^{\tau \wedge T} \mathscr{Z}_u\left( \frac{B_u^f}{B_u^r}\sum_{i=1}^d    Z_u^i dW_u^{f,i}\right)\Big \vert \mathcal{G}_t\right]} - \frac{1}{\mathscr{Z}_t} {\mathbb{E}_{\mathbb{P}}\left[  \int_t^{\tau \wedge T} \frac{B_u^f}{B_u^r}\mathscr{Z}_udH_u^{\mathcal{L}} \Big \vert \mathcal{G}_t\right]} \nonumber \\
   & \quad -\frac{1}{\mathscr{Z}_t} {\mathbb{E}_{\mathbb{P}}\left[ \int_t^{\tau \wedge T}  \mathscr{Z}_u \left( \frac{1}{B_u^r} dA_u^{R,C}+ \left( r_u^f-r_u\right)\frac{{V}_u^p(\varphi, A^{R,C})}{B_u^r}du \right)\Big \vert \mathcal{G}_t\right]} \nonumber \\
    & \quad -\frac{1}{\mathscr{Z}_t} {\mathbb{E}_{\mathbb{P}}\left[ \int_t^{\tau \wedge T} \widehat{V}_u^p(\varphi, A^{R,C})\mathscr{Z}_u\left( \sum_{i=1}^d -\frac{\mu^i(u,S_u)-r_u^i}{\sigma^{i}(u,S_u)} dW_u^{f,i} \right)\Big \vert \mathcal{G}_t\right]} \nonumber \\
    & \quad -  \frac{1}{\mathscr{Z}_t} {\mathbb{E}_{\mathbb{P}}\left[ \int_t^{\tau \wedge T} \frac{1}{B_u^r} d \langle C, \mathscr{Z}\rangle_u\Big \vert \mathcal{G}_t\right].} \label{eq:ConditionalExpectations}
    \end{align}
  We define the following processes by
  \begin{align}
    \mathscr{M}^{1,i}_t&:= \int_0^{t}  \frac{B_u^f}{B_u^r}\mathscr{Z}_u  Z_u^i dW_u^{f,i}, \quad \quad \mathscr{M}^{2}_t:=\int_0^{t}  \frac{B_u^f}{B_u^r} \mathscr{Z}_u d\mathcal{H}_u^{\mathcal{L}},\quad i=1,...,d, \, t \in [0, T], \nonumber  \\
     \mathscr{M}^{3,i}_t&:=\int_0^{t} -\widehat{V}_u^p(\varphi, A^{R,C})\mathscr{Z}_u\frac{\mu^i(u,S_u)-r_u^i}{\sigma^{i}(u,S_u)} dW_u^{\tau,f,i} , \quad i=1,...,d, \, t \in [0, T], \nonumber
  \end{align}
    and prove that they are martingales.
 To do so, we apply Theorem I.51 in \cite{protter}, i.e. we show that $\mathscr{M}$ for $\mathscr{M}\in \lbrace \mathscr{M}^{1,i},\mathscr{M}^{2}, \mathscr{M}^{3,i}\rbrace$ with $i=1,...,d$ is a local $(\mathbb{G}, \mathbb{P})$-martingale which satisfies
  \begin{align*}
    \mathbb{E}_{\mathbb{P}}\left[ \sup_{t \in [0, T]} \left \vert \mathscr{M}_t \right\vert\right] <\infty.
  \end{align*} 
  For $i=1,...,d$ we get for $\mathscr{M}^{1,i}$
  \begin{align*}
    \int_{0}^T \left \vert  \frac{B_u^f}{B_u^r} \mathscr{Z}_u  Z_u^i \right \vert^2 du \leq T\sup_{t \in[0,T]} \left \vert Z_t^i \right \vert^2 \sup_{t\in[0,T]} \left \vert \mathscr{Z}_t \right \vert^2 \sup_{t \in [0,T]} \left \vert  \frac{B_t^f}{B_t^r}\right \vert^2  < \infty \quad \mathbb{P}\text{-a.s.}, 
  \end{align*}
  as $Z \in \mathcal{H}^{2,d}(\mathbb{G}, \mathbb{P})$, by \eqref{eq:BoundsInterestRates} and by Doob's maximum inequality it follows from \eqref{eq:SecondMoment}
\begin{align}\label{eq:Doob}
\mathbb{E}_{\mathbb{P}}\left[ \sup_{t \in[0,T]}\vert \mathscr{Z}_t \vert^2\right] \leq \left(\frac{2}{2-1}\right)^2 \mathbb{E}_{\mathbb{P}}\left[ \vert \mathscr{Z}_T \vert^2\right]< \infty.
\end{align} 
Thus, $\mathscr{M}^{1,i}$ is a local $(\mathbb{G}, \mathbb{P})$-martingale.
  Furthermore, it holds
  \begin{align}
    \mathbb{E}_{\mathbb{P}}\left[ \sup_{t \in [0,T]} \left\vert \int_0^{t}  \frac{B_u^f}{B_u^r}\mathscr{Z}_u  Z_u^i dW_u^{f,i}\right \vert \right] &\leq K \mathbb{E}_{\mathbb{P}}\left[  \left \vert \int_0^{T}\left \vert  \frac{B_u^f}{B_u^r}\mathscr{Z}_u  Z_u^i\right\vert^2 du\right\vert^{\frac{1}{2}} \right] \label{eq:Burkholder-Davis-Gundy} \\
    & \leq K T \mathbb{E}_{\mathbb{P}}\left[  \sup_{t \in [0,T]} \left \vert Z_t^i \right\vert  \sup_{t \in [0,T]} \left \vert \mathscr{Z}_t \right\vert \sup_{t \in [0,T]} \left \vert  \frac{B_t^f}{B_t^r}\right \vert  \right] \nonumber \\
     & \leq K C(K_r) T \mathbb{E}_{\mathbb{P}}\left[  \sup_{t \in [0,T]} \left \vert Z_t^i \right \vert^2 \right]^{\frac{1}{2}} \mathbb{E}_{\mathbb{P}}\left[  \sup_{t \in [0,T]} \left \vert \mathscr{Z}_t \right \vert^2 \right]^{\frac{1}{2}}  \label{eq:Cauchy}  \\
     &< \infty,\label{eq:SecondMomentBound}
  \end{align}
  where we use in \eqref{eq:Burkholder-Davis-Gundy} the Burkholder-Davis-Gundy inequality with a suitable constant $K$. We get \eqref{eq:Cauchy} by the Cauchy-Schwarz inequality. Moreover, \eqref{eq:SecondMomentBound} follows by \eqref{eq:Doob} and $Z \in \mathcal{H}^{2,d}(\mathbb{G}, \mathbb{P})$ .
  Note that 
    \begin{align*}
        \mathscr{M}^{2}_t&=\int_0^{t} \frac{B_u^f}{B_u^r}\mathscr{Z}_u d\mathcal{H}_u^{\mathcal{L}}= \sum_{j \in \lbrace B,C \rbrace} \int_0^t  \frac{B_u^f}{B_u^r}\mathscr{Z}_u U_u^j dM^j_u,
  \end{align*}
  and $[ M^j ]_u= H_u^j=\textbf{1}_{\lbrace \tau^j \leq u\rbrace}$. Thus, we get for $j \in \lbrace{ B, C \rbrace}$
    \begin{align*}
        \int_0^T  \left \vert  \frac{B_u^f}{B_u^r} \mathscr{Z}_u U_u^j\right \vert^2 d [ M^j ]_u&= \int_0^T  \left \vert  \frac{B_u^f}{B_u^r} \mathscr{Z}_u U_u^j\right \vert^2 d\textbf{1}_{\lbrace \tau^j \leq u\rbrace} \nonumber \\
        &= C(K_r)\left(\mathscr{Z}_{\tau^j} U_{\tau^j}^j\right)^2  \textbf{1}_{\lbrace \tau^j \leq T\rbrace} < \infty \quad \mathbb{P}\text{-a.s.}
    \end{align*}
    and 
    \begin{align}
        \mathbb{E}_{\mathbb{P}} \left[\sup_{t \in [0,T]} \left \vert  \int_0^t \frac{B_u^f}{B_u^r}\mathscr{Z}_u U_u^j dM^j_u \right \vert \right] &\leq K \mathbb{E}_{\mathbb{P}} \left[\left \vert  \int_0^T \left( \frac{B_u^f}{B_u^r}\mathscr{Z}_u U_u^j \right)^2 d[ M^j]_u \right \vert^{\frac{1}{2}} \right] \nonumber \\
        &= K \mathbb{E}_{\mathbb{P}} \left[ \frac{B_{\tau^j}^f}{B_{\tau^j}^r}\mathscr{Z}_{\tau^j} U_{\tau^j}^j \textbf{1}_{\lbrace \tau^j \leq T\rbrace} \right] \nonumber \\
        & \leq KC(K_r) \mathbb{E}_{\mathbb{P}} \left[ \mathscr{Z}_{\tau^j}^2  \textbf{1}_{\lbrace \tau^j \leq T\rbrace} \right]^{\frac{1}{2}}\mathbb{E}_{\mathbb{P}} \left[ (U_{\tau^j}^j)^2  \textbf{1}_{\lbrace \tau^j \leq T\rbrace} \right]^{\frac{1}{2}}< \infty, \nonumber
    \end{align}
    by \eqref{eq:Doob} and $U \in \mathcal{H}_{\lambda}^{2,2}(\mathbb{G},\mathbb{P}).$

 It remains to prove that $\mathscr{M}^{3,i}$ is a $(\mathbb{G}, \mathbb{P})$-martingale. For $i=1,...,d$ we have
  \begin{align}
  &\int_0^{\tau \wedge T}\left \vert \widehat{V}_u^p(\varphi, A^{R,C})\mathscr{Z}_u\frac{\mu^i(u,S_u)-r_u^i}{\sigma^{i}(u,S_u)}\right \vert ^2 du \nonumber \\
  &\leq  K_i^2 \sup_{t \in [0,\tau \wedge T]} \vert \mathscr{Z}_t \vert^2    \int_0^{\tau \wedge T}\left \vert \widehat{V}_u^p(\varphi, A^{R,C})\right\vert^2 du \nonumber \\
  &\leq K_i^2 \sup_{t \in [0,T]} \vert \mathscr{Z}_t \vert^2  \left(  2\int_0^{\tau \wedge T}\left \vert \frac{\psi_u^f B^f_u}{B_u^r} \right\vert^2 du + 2\int_0^{\tau \wedge T}\left \vert \frac{C_u}{B_u^r}  \right\vert^2 du \right)\label{eq:DefinitionValueProcess1} \\
   &\leq K_i^2 \sup_{t \in [0,T]} \vert \mathscr{Z}_t \vert^2  \left(  2\int_0^{\tau \wedge T}\left \vert \frac{ B^f_u}{B_u^r} \right \vert^2 \left \vert Y_u + \int_{0}^u \frac{1}{B_s^f} d A^{R,C} _s\right\vert^2 du + 2 T \sup_{t \in[0,T]} \left \vert C_t \right \vert^2 \sup_{t \in[0,T]} \left \vert \frac{1}{B_t^r} \right \vert^2   \right)\label{eq:RepresentationBSDE}  \\
    &\leq K_i^2 C(K_r) \sup_{t \in [0,T]}  \vert \mathscr{Z}_t \vert^2  \left(\sup_{t \in [0,T]}  \left(  4 T \sup_{t \in  [0,\tau \wedge T]}\left \vert Y_t \right \vert^2   + 4 \int_{0}^{\tau \wedge T} \left \vert \int_0^u \frac{1}{B_s^f} d A^{R,C}_s \right\vert^2 du\right) + 2 T \sup_{t \in[0,T]} \left \vert C_t \right \vert^2   \right)\label{eq:BoundedTreasuryAccountUse} \\
    % &= K_i^2 \sup_{t \in [0,T]} \vert \mathscr{Z}_t \vert^2  \left(  4 T \sup_{t \in  [0,T]}\left \vert Y_t \right \vert^2 + 4 \int_{0}^T \left \vert \int_0^u \frac{1}{B_s^f} d A^{R,C}_s \right\vert^2 du + 2 T \sup_{t \in[0,T]} \left \vert C_t \right \vert^2 \right)\nonumber \\
    & < \infty, \label{eq:MartingaleEstimates}
  \end{align}
  where we use in \eqref{eq:DefinitionValueProcess1} that 
  \begin{align}
    \widehat{V}_u^p(\varphi, A^{R,C})= \frac{ {V}_u^p(\varphi, A^{R,C})}{B_u^r} + \frac{C_u}{B_u^r}=\frac{\psi_u^f B_u^f}{B_u^r}+ \frac{C_u}{B_u^r}, \nonumber
  \end{align}
   and \eqref{eq:RepresentationBSDE} follows by Lemma \ref{lemma:DescriptionStrategyBSDE}. Finally, we conclude that \eqref{eq:MartingaleEstimates} holds by \eqref{eq:IntegrabilityCollateral}, \eqref{eq:Doob}, \eqref{eq:AdditionalAssumption1} and $Y \in \mathbb{S}^2(\mathbb{G}, \mathbb{P})$.
  With similar arguments we get for $i=1,...,d$
  \begin{align}
   & \mathbb{E}_{\mathbb{P}}\left[ \sup_{t \in [0, T]}  \left \vert \mathscr{M}_t^{3,i}  \right \vert \right] \nonumber \\
   &= \mathbb{E}_{\mathbb{P}}\left[ \sup_{t \in [0, T]} \left \vert \int_0^{t} \widehat{V}_u^p(\varphi, A^{R,C})\mathscr{Z}_u\frac{\mu^i(u,S_u)-r_u^i}{\sigma^{i}(u,S_u)} dW_u^{f,i}  \right \vert \right] \nonumber \\
   & \leq K K_i^2\mathbb{E}_{\mathbb{P}}\left[  \left \vert \int_0^{{\tau \wedge T}} \left \vert \widehat{V}_u^p(\varphi, A^{R,C})\mathscr{Z}_u \right \vert^2 d u \right \vert^{\frac{1}{2}} \right] \nonumber  \\
    & \leq K K_i^2\mathbb{E}_{\mathbb{P}}\left[ \sup_{t \in [0,T]} \vert \mathscr{Z}_t \vert  \left \vert \int_0^{{\tau \wedge T}} \left \vert \widehat{V}_u^p(\varphi, A^{R,C}) \right \vert^2 d u \right \vert^{\frac{1}{2}} \right]\nonumber \\
   & \leq K K_i^2\mathbb{E}_{\mathbb{P}}\left[ \sup_{t \in [0,T]} \vert \mathscr{Z}_t \vert^2 \right]  ^{\frac{1}{2}} \mathbb{E}_{\mathbb{P}}\left[  \int_0^{{\tau \wedge T}} \left \vert \widehat{V}_u^p(\varphi, A^{R,C}) \right \vert^2 d u  \right]^{\frac{1}{2}} \nonumber \\
   & \leq 4 K K_i^2 T C(K_r) \mathbb{E}_{\mathbb{P}}\left[ \sup_{t \in [0,T]} \vert \mathscr{Z}_t \vert^2 \right]^{\frac{1}{2}} \mathbb{E}_{\mathbb{P}}\left[ \sup_{t \in  [0,{\tau \wedge T}]}\left \vert Y_t  \right \vert^2 + \int_{0}^{{\tau \wedge T}} \left \vert \int_0^u \frac{1}{B_s^f} d A^{R,C}_s \right\vert^2 du +  \sup_{t \in[0,T]} \left \vert C_t \right \vert^2  \right]^{\frac{1}{2}}\nonumber \\
    & \leq 4 K K_i^2 T C(K_r) \mathbb{E}_{\mathbb{P}}\left[ \sup_{t \in [0,T]} \vert \mathscr{Z}_t \vert^2 \right]^{\frac{1}{2}} \Bigg (\mathbb{E}_{\mathbb{P}}\left[   \sup_{t \in  [0,{\tau \wedge T}]}\left \vert Y_t \right \vert^2 \right] + \mathbb{E}_{\mathbb{P}}\left[\int_{0}^{{\tau \wedge T}} \left \vert \int_0^u \frac{1}{B_s^f} d A^{R,C}_s \right\vert^2 du \right]\nonumber \\
    & \quad +\mathbb{E}_{\mathbb{P}}\left[ \sup_{t \in[0,T]} \left \vert C_t \right \vert^2  \right] \Bigg)^{\frac{1}{2}}\nonumber \\
    & \quad < \infty. \nonumber 
  \end{align}
  Thus, \eqref{eq:ConditionalExpectations} simplifies to 
   \begin{align}
  \widehat{V}_t^p(\varphi, A^{R,C})&= \widehat{V}_t^p(\varphi, {A}^{C})\nonumber \\
  &= \frac{1}{\mathscr{Z}_t} \mathbb{E}_{\mathbb{P}}\left[ \left(\frac{C_T}{B_T^r} \textbf{1}_{\lbrace \tau >T \rbrace}- \frac{R_{\tau}}{B_{\tau}^r}\textbf{1}_{\lbrace \tau \leq T \rbrace}\right) \mathscr{Z}_{\tau \wedge T} \Big \vert \mathcal{G}_t\right] -  \frac{1}{\mathscr{Z}_t} {\mathbb{E}_{\mathbb{P}}\left[ \int_t^{\tau \wedge T} \frac{1}{B_u^r} d \langle C, \mathscr{Z}\rangle_u\Big \vert \mathcal{G}_t\right]}\nonumber \\
  & \quad -\frac{1}{\mathscr{Z}_t} {\mathbb{E}_{\mathbb{P}}\left[ \int_t^{\tau \wedge T} \mathscr{Z}_u \left( \frac{1}{B_u^r} dA_u^{C}+ \left( r_u^f-r_u\right)\frac{{V}_u^p(\varphi, A^{C})}{B_u^r}du \right)\Big \vert \mathcal{G}_t\right]}, \label{eq:SplittingA}
   \end{align}
  where we use that on $\lbrace{t < \tau \rbrace}$ it holds $ \widehat{V}_t^p(\varphi, {A}^{C}) =\widehat{V}_t^p(\varphi, A^{R,C})$ as well as $A_t^{R,C}=A_t^C$.\\

  Next, we rewrite the term $\int_{t}^{\tau \wedge T} \mathscr{Z}_u\frac{1}{B_u^r}dA_u^{C}$. To simplify the notation we set 
  \begin{align*} 
\widetilde{A}_t:=  \textbf{1}_{\lbrace t < \tau  \rbrace} A_t +  \textbf{1}_{\lbrace t \geq \tau  \rbrace} A_{\tau-}, \quad \widetilde{C}_t:=\textbf{1}_{\lbrace t < \tau  \rbrace} C_t + \textbf{1}_{\lbrace t \geq \tau  \rbrace}C_{\tau-}, \quad t \in [0, \tau \wedge T],
    \end{align*}
    and thus by \eqref{eq:CashflowDefaultCollateralized} and \eqref{eq:CashflowDefaultCollateralized1} we have 
    \begin{align}\label{eq:CashflowAdjusted}
        {A}_t^C:=\widetilde{A}_t+\widetilde{C}_t-\int_0^t \bar{r}_u^c C_u du.
    \end{align} It holds
    \begin{align}
    \int_{t}^{\tau \wedge T} \frac{\mathscr{Z}_u}{B_u^r}d \widetilde{C}_u&=\frac{\mathscr{Z}_{\tau \wedge T}C_{\tau \wedge T}}{B_{\tau \wedge T}^r}-\frac{\mathscr{Z}_tC_t}{B_t^r}-\int_t^{\tau \wedge T} \mathscr{Z}_uC_u d(B_u^r)^{-1}- \int_t^{\tau \wedge T} \frac{C_u}{B_u^r} d\mathscr{Z}_u- \int_t^{\tau \wedge T} \frac{1}{B_u^r} d \langle C, \mathscr{Z} \rangle_u
    \label{eq:Ito1} \\
    &=\frac{\mathscr{Z}_{\tau \wedge T}C_{\tau \wedge T}}{B_{\tau \wedge T}^r}-\frac{\mathscr{Z}_tC_t}{B_t^r}+ \int_t^{\tau \wedge T} \frac{ \mathscr{Z}_u C_u r_u}{B_u^r} du- \int_t^{\tau \wedge T} \frac{C_u}{B_u^r} d\mathscr{Z}_u - \int_t^{\tau \wedge T} \frac{1}{B_u^r} d \langle C, \mathscr{Z} \rangle_u \label{eq:Ito2}, 
    \end{align} 
where we use in \eqref{eq:Ito1} that by It\^{o}'s formula
\begin{align*}
 d\frac{\mathscr{Z}_t C_t}{B_t^r}= \mathscr{Z}_t C_t d(B_t^r)^{-1}+ (B_t^r)^{-1} d(\mathscr{Z}_tC_t) = \mathscr{Z}_t C_t d(B_t^r)^{-1}+ (B_t^r)^{-1} (\mathscr{Z}_tdC_t+ C_t d\mathscr{Z}_t + {d \langle C,\mathscr{Z}  \rangle_t}). 
\end{align*}
Combining \eqref{eq:SplittingA}, \eqref{eq:CashflowAdjusted} and \eqref{eq:Ito2} yields
\begin{align}
 \widehat{V}_t^p(\varphi, {A}^{R,C})&= \frac{1}{\mathscr{Z}_t} \mathbb{E}_{\mathbb{P}}\left[ \left(\frac{C_T}{B_T^r} \textbf{1}_{\lbrace \tau >T \rbrace}- \frac{R_{\tau}}{B_{\tau}^r}\textbf{1}_{\lbrace \tau \leq T \rbrace}\right) \mathscr{Z}_{\tau \wedge T} \Big \vert \mathcal{G}_t\right] -  \frac{1}{\mathscr{Z}_t} {\mathbb{E}_{\mathbb{P}}\left[ \int_t^{\tau \wedge T} \frac{1}{B_u^r} d \langle C, \mathscr{Z}\rangle_u\Big \vert \mathcal{G}_t\right]}\nonumber \\
  & \quad -\frac{1}{\mathscr{Z}_t} {\mathbb{E}_{\mathbb{P}}\left[ \int_t^{\tau \wedge T} \mathscr{Z}_u \left( \frac{1}{B_u^r} d(\widetilde{A}_u+\widetilde{C}_u- \bar{r}_u^c C_u du)+ \left( r_u^f-r_u\right)\frac{{V}_u^p(\varphi, A^{C})}{B_u^r}du \right)\Big \vert \mathcal{G}_t\right]}, \nonumber \\
 &=\frac{1}{\mathscr{Z}_t} \mathbb{E}_{\mathbb{P}}\left[ \left(\frac{C_T}{B_T^r} \textbf{1}_{\lbrace \tau >T \rbrace}- \frac{R_{\tau}}{B_{\tau}^r}\textbf{1}_{\lbrace \tau \leq T \rbrace}\right) \mathscr{Z}_{\tau \wedge T} - \int_{t}^{\tau \wedge T} \frac{\mathscr{Z}_u}{B_u^r}d\widetilde{A}_u\Big \vert \mathcal{G}_t\right] \nonumber \\
 & \quad -\frac{1}{\mathscr{Z}_t} \mathbb{E}_{\mathbb{P}}\left[ \frac{\mathscr{Z}_{\tau \wedge T}C_{\tau \wedge T}}{B_{\tau \wedge T}^r}-\frac{\mathscr{Z}_tC_t}{B_t^r}+ \int_t^{\tau \wedge T} \frac{ \mathscr{Z}_u C_u r_u}{B_u^r} du- \int_t^{\tau \wedge T} \frac{C_u}{B_u^r} d\mathscr{Z}_u\Big \vert \mathcal{G}_t\right] \nonumber \\
 &\quad  + \frac{1}{\mathscr{Z}_t} \mathbb{E}_{\mathbb{P}}\left[  \int_t^{\tau \wedge T} \frac{ \mathscr{Z}_u \bar{r}_u^c}{B_u^r} du- \int_t^{\tau \wedge T} \mathscr{Z}_u\left( r_u^f-r_u\right)\frac{{V}_u^p(\varphi, A^{C})}{B_u^r}du\Big \vert \mathcal{G}_t\right].\label{eq:IntermediateStepProoof}
\end{align}
By Definition \ref{def:ValueProcessWealth}, \eqref{eq:NoCollateralPosted} and \eqref{eq:IntermediateStepProoof} it follows
\begin{align}
 \widehat{V}_t(\varphi, {A}^{R,C})& 
 = \widehat{V}_t^p(\varphi, A^{R,C}) -\frac{C_t}{B_t^r}\nonumber \\
 &=\frac{1}{ \mathscr{Z}_t}\mathbb{E}_{\mathbb{P}}\left[ -\mathscr{Z}_{\tau \wedge T} \frac{R_{\tau}}{B_{\tau}^r}\textbf{1}_{\lbrace \tau \leq T \rbrace}  - \int_{t}^{\tau \wedge T} \frac{\mathscr{Z}_u}{B_u^r}d\widetilde{A}_u\Big \vert \mathcal{G}_t\right] \nonumber \\
 & \quad - \frac{1}{ \mathscr{Z}_t}\mathbb{E}_{\mathbb{P}}\left[ \mathscr{Z}_{\tau \wedge T} \frac{C_{\tau-}}{B_{\tau}^f} \textbf{1}_{\lbrace  \tau \leq T \rbrace}+ \int_t^{\tau \wedge T} \frac{ \mathscr{Z}_u C_u r_u}{B_u} du- \int_t^{\tau \wedge T} \frac{C_u}{B_u} d\mathscr{Z}_u\Big \vert \mathcal{G}_t\right] \nonumber \\
  & \quad + \frac{1}{\mathscr{Z}_t} \mathbb{E}_{\mathbb{P}}\left[  \int_t^{\tau \wedge T} \frac{ \mathscr{Z}_u \bar{r}_u^c}{B_u^r} du- \int_t^{\tau \wedge T} \mathscr{Z}_u\left( r_u^f-r_u\right)\frac{{V}_u^p(\varphi, A^{C})}{B_u^r}du\Big \vert \mathcal{G}_t\right].  \label{eq:IntermediateStepProof2}
\end{align}
Furthermore, it holds 
\begin{align*}
\mathbb{E}_{\mathbb{P}}\left[  \int_t^{\tau \wedge T} \frac{C_u}{B_u^r} d\mathscr{Z}_u\Big \vert \mathcal{G}_t\right]=\mathbb{E}_{\mathbb{P}}\left[  \int_t^{\tau \wedge T} \frac{C_u}{B_u^r} \mathscr{Z}_u \left( \sum_{i=1}^d -\frac{\mu^i(u,S_u)-r_u^i}{\sigma^{i}(u,S_u)} d W_u^{f,i}\right)  \Big \vert \mathcal{G}_t\right],
\end{align*}
and we now show that 
$$
\mathscr{M}_t^{4,i}:= - \int_0^t \frac{C_u}{B_u^r} \mathscr{Z}_u \frac{\mu^i(u,S_u)-r_u^i}{\sigma^{i}(u,S_u)} dW_u^{\tau,f,i}, \quad i=1...,d, \, t \in [0,T]
$$
is a $(\mathbb{G},\mathbb{P})$-martingale. In particular, we have
\begin{align*}
    \int_0^{\tau \wedge T}  \left \vert \frac{C_u}{B_u^r} \mathscr{Z}_u \frac{\mu^i(u,S_u)-r_u^i}{\sigma^{i}(u,S_u)}\right \vert^2 du &\leq K_i^2 \int_0^T \left \vert \frac{C_u}{B_u^r}  \mathscr{Z}_u \right \vert^2 du \\
    &\leq K_i^2 T \sup_{t \in [0,T]} \vert C_t \vert^2  \sup_{t \in [0,T]} \left \vert  \mathscr{Z}_t \right \vert^2 \sup_{t \in [0,T]} \left \vert  \frac{1}{B_t^r} \right \vert^2 < \infty \quad \mathbb{P} \text{-a.s.}
\end{align*}
and
\begin{align}
\mathbb{E}_{\mathbb{P}}\left[ \sup_{t \in [0,T]} \left \vert \mathscr{M}_t^{4,i} \right \vert \right]&=\mathbb{E}_{\mathbb{P}}\left[ \sup_{t \in [0,T]} \left \vert \int_0^t \frac{C_u}{B_u^r} \mathscr{Z}_u \frac{\mu^i(u,S_u)-r_u^i}{\sigma^{i}(u,S_u)} dW_u^{\tau,f,i} \right \vert \right] \nonumber \\
&\leq K K_i^2 \mathbb{E}_{\mathbb{P}}\left[ \left \vert \int_0^{\tau \wedge T} \frac{1}{\vert B_u^r\vert^2} \vert C_u\vert^2 \vert \mathscr{Z}_u \vert^2  du\right \vert^{\frac{1}{2}} \right] \nonumber \\
&\leq K  K_i^2 T  C(K_r) \mathbb{E}_{\mathbb{P}}\left[ \sup_{t \in [0,T]}\vert C_t \vert   \sup_{t \in [0,T]} \vert \mathscr{Z}_t \vert   \right] \nonumber \\
&\leq K K_i^2  C(K_r) \mathbb{E}_{\mathbb{P}}\left[ \sup_{t \in [0,T]}\vert C_t \vert^2  \right]^{\frac{1}{2}}  \mathbb{E}_{\mathbb{P}} \left[  \sup_{t \in [0,T]} \vert \mathscr{Z}_t \vert^2 \right]^{\frac{1}{2}} \nonumber \\
&< \infty \nonumber 
\end{align}
by using \eqref{eq:IntegrabilityCollateral} and \eqref{eq:Doob} and similar arguments as before. This implies that $\mathbb{E}_{\mathbb{P}}\left[  \int_t^{\tau \wedge T} \frac{C_u}{B_u^r} d\mathscr{Z}_u\Big \vert \mathcal{G}_t\right]=0.$
Thus, by \eqref{eq:IntermediateStepProof2} we get
\begin{align}
\widehat{V}_t(\varphi, {A}^{R,C}) &=\frac{1}{ \mathscr{Z}_t}\mathbb{E}_{\mathbb{P}}\left[ - \frac{\mathscr{Z}_{\tau }}{B_{\tau}^r}( C_{\tau-}+R_{\tau})\textbf{1}_{\lbrace \tau \leq T \rbrace}  - \int_{t}^{\tau \wedge T}\frac{ \mathscr{Z}_u}{B_u^r}d\widetilde{A}_u + \int_{t}^{\tau \wedge T} \mathscr{Z}_u\frac{C_u}{B_u^r}\left(\bar{r}_u^c-{r}_u \right) du\Big \vert \mathcal{G}_t\right] \nonumber \\ 
&\quad -\frac{1}{\mathscr{Z}_t}\mathbb{E}_{\mathbb{P}}\left[\int_t^{\tau \wedge T} \mathscr{Z}_u\left( r_u^f-r_u\right)\frac{{V}_u^p(\varphi, A^{C})}{B_u^r}du\Big \vert \mathcal{G}_t\right] \nonumber \\
 &=\frac{1}{ \mathscr{Z}_t}\mathbb{E}_{\mathbb{P}}\left[ - \frac{\mathscr{Z}_{\tau }}{B_{\tau}^r}( C_{\tau-}+R_{\tau})\textbf{1}_{\lbrace \tau \leq T \rbrace}  - \int_{t}^{\tau \wedge T}\frac{ \mathscr{Z}_u}{B_u^r}d\widetilde{A}_u + \int_{t}^{\tau \wedge T} \mathscr{Z}_u\frac{C_u}{B_u^r}\left(\bar{r}_u^c-{r}_u \right) du\Big \vert \mathcal{G}_t\right] \nonumber \\ 
&\quad -\frac{1}{ \mathscr{Z}_t}\mathbb{E}_{\mathbb{P}}\left[\int_t^{\tau \wedge T} \mathscr{Z}_u\left( r_u^f-r_u\right)\frac{{V}_u(\varphi, A^{C}) + C_u}{B_u^r}du\Big \vert \mathcal{G}_t\right] \nonumber \\
&=\frac{1}{ \mathscr{Z}_t}\mathbb{E}_{\mathbb{P}}\left[ - \frac{\mathscr{Z}_{\tau }}{B_{\tau}^r}( C_{\tau-}+R_{\tau})\textbf{1}_{\lbrace \tau \leq T \rbrace}  - \int_{t}^{\tau \wedge T}\frac{ \mathscr{Z}_u}{B_u^r}d\widetilde{A}_u + \int_{t}^{\tau \wedge T} \mathscr{Z}_u\frac{C_u}{B_u^r}\left(\bar{r}_u^c-{r}_u^f \right) du\Big \vert \mathcal{G}_t\right] \nonumber \\ 
&\quad -\frac{1}{ \mathscr{Z}_t}\mathbb{E}_{\mathbb{P}}\left[\int_t^{\tau \wedge T} \mathscr{Z}_u\left( r_u^f-r_u\right)\frac{{V}_u(\varphi, A^{C})}{B_u^r}du\Big \vert \mathcal{G}_t\right], \label{eq:RepresentationXVA1}
\end{align}
where we use that ${V}_u^p(\varphi, A^{C}) ={V}_u(\varphi, A^{C}) + C_u$.
Next, we study the term $-\frac{1}{ \mathscr{Z}_t}\mathbb{E}_{\mathbb{P}}\left[ \int_{t}^{\tau \wedge T} \frac{ \mathscr{Z}_u}{B_u^r}d\widetilde{A}_u\Big \vert \mathcal{G}_t\right]$. Note that on $\lbrace t < \tau  \rbrace$ it holds
\begin{align*}
 \int_{]t,\tau \wedge T]}  \frac{\mathscr{Z}_u}{B_u^r}d\widetilde{A}_u=\int_{]t, T]}  \frac{\mathscr{Z}_u}{B_u^r}d{A}_u - \textbf{1}_{\lbrace \tau \leq T \rbrace} \int_{[\tau \wedge T,T]}  \frac{\mathscr{Z}_u}{B_u^r}d{A}_u
\end{align*}
and thus
\begin{align}
\frac{B_t^r}{ \mathscr{Z}_t}\mathbb{E}_{\mathbb{P}}\left[-\int_{t}^{\tau \wedge T} \frac{\mathscr{Z}_u}{B_u^r}d\widetilde{A}_u\Big \vert \mathcal{G}_t\right]& = \frac{B_t^r}{ \mathscr{Z}_t} \mathbb{E}_{\mathbb{P}}\left[-\int_{]t, T]} \frac{ \mathscr{Z}_u}{B_u^r}d{A}_u + \textbf{1}_{\lbrace \tau \leq T \rbrace} \int_{[\tau \wedge T,T]}  \mathscr{Z}_u\frac{1}{B_u^r}d{A}_u  \Big \vert \mathcal{G}_t\right] \nonumber \\
&=\mathcal{V}_t + \frac{B_t^r}{ \mathscr{Z}_t}\mathbb{E}_{\mathbb{P}}\left[ \mathbb{E}_{\mathbb{P}} \left[\textbf{1}_{\lbrace \tau \leq T \rbrace}\frac{\mathscr{Z}_{\tau}}{\mathscr{Z}_{\tau}}\frac{{B}_{\tau}^r}{{B}_{\tau}^r} \int_{[\tau \wedge T,T]}  \frac{ \mathscr{Z}_u}{B_u^r}d{A}_u \Big \vert \mathcal{G}_{\tau}\right]\Big \vert \mathcal{G}_t\right] \nonumber \\
&=\mathcal{V}_t + {\frac{B_t^r}{ \mathscr{Z}_t}\mathbb{E}_{\mathbb{P}}\left[ \textbf{1}_{\lbrace \tau \leq T \rbrace}\frac{\mathscr{Z}_{\tau}}{B_{\tau}^r}Q_{\tau} \Big \vert \mathcal{G}_t\right] } \label{eq:DefCleanValue}
\end{align}
where we use Remark \ref{remark:ConditioningGF} with $\mathcal{V}_t$ and $Q_{\tau}$ given in Definition \ref{def:CleanValue}. 
 Combining \eqref{eq:RepresentationXVA1} and \eqref{eq:DefCleanValue} yields
\begin{align}
{V}_t(\varphi, {A}^{R,C}) &= \widehat{V}_t(\varphi, {A}^{R,C}) B_t^r\nonumber \\
&=\mathcal{V}_t+ {\frac{B_t^r}{ \mathscr{Z}_t}\mathbb{E}_{\mathbb{P}}\left[ \textbf{1}_{\lbrace \tau \leq T \rbrace} \textbf{1}_{\lbrace{\tau^C< \tau^B\rbrace}}\frac{\mathscr{Z}_{\tau}}{B_{\tau}^r}(1-R^C) (Q_{\tau}-C_{\tau-})^+\Big \vert \mathcal{G}_t\right] }  \nonumber \\
& \quad - \frac{B_t^r}{ \mathscr{Z}_t}\mathbb{E}_{\mathbb{P}}\left[ \textbf{1}_{\lbrace \tau \leq T \rbrace}\textbf{1}_{\lbrace{\tau^B< \tau^C\rbrace}}\frac{\mathscr{Z}_{\tau}}{B_{\tau}^r} (1-R^B)(Q_{\tau}-C_{\tau-})^-\Big \vert \mathcal{G}_t\right] \nonumber \\
& \quad + \frac{B_t^r}{ \mathscr{Z}_t}\mathbb{E}_{\mathbb{P}}\left[ \int_{t}^{\tau \wedge T} \mathscr{Z}_u\frac{C_u}{B_u^r}\left(\bar{r}_u^c-r_u^f \right) du\Big \vert \mathcal{G}_t\right] \nonumber \\
&\quad - \frac{B_t^r}{ \mathscr{Z}_t}\mathbb{E}_{\mathbb{P}}\left[ \int_{t}^{\tau \wedge T} \mathscr{Z}_u\frac{V_u(\varphi,A^{R,C})}{B_u^r}\left({r}_u^f-r_u \right) du\Big \vert \mathcal{G}_t\right]\nonumber \\
&= \mathcal{V}_t+ \textnormal{CVA}_t-\textnormal{DVA}_t+\textnormal{CollVA}_t -\textnormal{FVA}_t, \nonumber 
\end{align}
and thus
\begin{align}
{V}_t(\varphi, {A}^{R,C})&= \mathcal{V}_t+ \textnormal{CVA}_t-\textnormal{DVA}_t+\textnormal{CollVA}_t \nonumber \\
&  \quad - \frac{B_t^r}{ \mathscr{Z}_t}\mathbb{E}_{\mathbb{P}}\left[ \int_{t}^{\tau \wedge T} \mathscr{Z}_u\frac{\mathcal{V}_u+ \textnormal{CVA}_u-\textnormal{DVA}_u+\textnormal{CollVA}_u-\textnormal{FVA}_u}{B_u^r}\left({r}_u^f-r_u \right) du\Big \vert \mathcal{G}_t\right].  \nonumber 
\end{align}
\end{proof}

The result above provides the decomposition of the price in terms of valuation adjustments. To this extent we prove a consolidated practice in banks' risk management in a mathematically rigorous way also for the case of incomplete markets.

\section{Two-Step valuation of xVAs}
In this section we show that a slight modification of the assumptions underlying our approach based on local risk-minimization provides a foundation to a version of the two-step valuation approach proposed in \cite{pelsser_stadje_2014}, \cite{delong_dhaene_barigou_2019} and \cite{barigou_linders_yang_2022}. This allows us to obtain a formula for CVA (and similarly for other valuation adjustments) which is particularly convenient in applications.

We consider $W^f=(W_t^f)_{t \in [0,T]}$ and $W^\perp=(W_t^\perp)_{t \in [0,T]}$ to be $d$- resp. $d^\prime$-dimensional Brownian motions on $(\Omega, \mathcal{G}, \mathbb{P})$ and $\mathbb{F}^{W^f}=(\mathcal{F}^{W^f}_t)_{t \in [0,T]}$ and $\mathbb{F}^{W^\perp}=(\mathcal{F}^{W^\perp}_t)_{t \in [0,T]}$ their associated natural filtrations. We set $\mathbb{F}:=\mathbb{F}^{W^f}\otimes\mathbb{F}^{W^\perp}$, that is $\mathcal{F}_t:=\mathcal{F}^{W^f}_t\otimes\mathcal{F}^{W^\perp}_t$ for every $t\in[0,T]$. Let the stopping times $\tau^B, \tau^C, \tau$ and the associated filtrations be as introduced in Section  \ref{Sect:preliminaries}. On $(\Omega, \mathcal{G}, \mathbb{P})$ we consider the filtration $\mathbb{G}=(\mathcal{G}_t)_{t\in[0,T]}$ defined by
\begin{align*}
\mathcal{G}_t &:= \mathcal{F}_t \vee \mathcal{H}^B_t \vee \mathcal{H}^C_t\\
&=\mathcal{F}^{W^f}_t\otimes\mathcal{F}^{W^\perp}_t \vee \mathcal{H}^B_t \vee \mathcal{H}^C_t, \quad t \in [0,T],
\end{align*}
and set $\mathcal{G}:=\mathcal{G}_T.$ All filtrations are required to satisfy the usual hypotheses of completeness and right-continuity.

In addition to the assumptions in Section \ref{Sect:preliminaries}, we assume that all coefficients $\mu^i,\sigma^i$, repo rates $r^i$ for $i=1,...,d,$ collateral and funding rates are $\mathbb{F}^{W^f}$-adapted, so that completeness is preserved for the clean market. With respect to the previous sections, we assume that the close-out valuation process $Q$ is $\mathbb{F}^{W^f}$-adapted by setting
\begin{align*}
     Q_t:=  \mathbb{E}_{\mathbb{Q}}\left[\int_{[t,T]}  \frac{B_t^r}{B_u^r}\ d A_u \Big \vert \mathcal{F}^{W^f}_t \right]=\mathbb{E}_{\mathbb{P}}\left[\int_{[t,T]}  {\mathscr{Z}}_u \frac{B_t^r}{B_u^r}\frac{1}{ {\mathscr{Z}}_t} d A_u \Big \vert \mathcal{F}^{W^f}_t \right] \quad t \in [0,T].
\end{align*}
Similarly, we also assume that the collateral process $C$ is $\mathbb{F}^{W^f}$-adapted.

We allow the default intensities to be $\mathbb{F}$-adapted, which introduces an additional dependence on the Brownian motion $W^\perp$. Recalling our previous assumption that intensity rates are bounded, a concrete example that satisfies all the assumptions of the present setting is given by the Jacobi diffusion
\begin{align*}
    d\lambda^j_t=\kappa^j(\theta^j-\lambda^j_t)dt+\sqrt{(\lambda_t-\lambda^j_{max})(\lambda_t-\lambda^j_{min})}\left(\rho^jdW^{f,1}_t+\sqrt{1-(\rho^{j})^2}dW^{\perp,j}\right), \quad t \in [0,T],
\end{align*}
where $\lambda^j_{max},\lambda^j_{min}$ are the upper and lower bound for the intensity respectively, $\kappa^j,\theta^j\in \mathbb{R}_+$, $\rho^j\in(-1,1)$  for $j\in\{B,C\}$. In this setting, let us consider the CVA which can be rewritten as
\begin{align*}
    \textnormal{CVA}_t&=\frac{B_t^r}{ \mathscr{Z}_t}\mathbb{E}_{\mathbb{P}}\left[ \textbf{1}_{\lbrace \tau \leq T \rbrace} \textbf{1}_{\lbrace{\tau^C< \tau^B\rbrace}}\frac{\mathscr{Z}_{\tau}}{B_{\tau}^r}(1-R^C) (Q_{\tau}-C_{\tau-})^+\Big \vert \mathcal{G}_t\right]\\
    &=\textbf{1}_{\lbrace \tau > t \rbrace}\frac{B_t^r}{ \mathscr{Z}_t}\mathbb{E}_{\mathbb{P}}\left[  (1-R^C)\int_t^T\frac{\mathscr{Z}_{u}}{B_{u}^r}e^{-\int_t^u\lambda^C_s+\lambda^B_s ds}\lambda^C_u (Q_{u}-C_{u})^+du\Big \vert \mathcal{F}_t\right], \quad t \in [0,T],
\end{align*}
by using Proposition 5.1.1 in \cite{BieRut2004}. Next, following \cite{pelsser_stadje_2014}, we condition on the full information on the market risk factors up to maturity and consider the two-step valuation
\begin{align*}
    \textnormal{CVA}_t&=\textbf{1}_{\lbrace \tau > t \rbrace}\frac{B_t^r}{ \mathscr{Z}_t}\mathbb{E}_{\mathbb{P}}\left[  (1-R^C)\int_t^T\frac{\mathscr{Z}_{u}}{B_{u}^r}\mathbb{E}_{\mathbb{P}}\left[e^{-\int_t^u\lambda^C_s+\lambda^B_s ds}\lambda^C_u \Big \vert \mathcal{F}^{W^f}_T\otimes\mathcal{F}^{W^\perp}_t\right](Q_{u}-C_{u})^+du\Big \vert \mathcal{F}_t\right],
\end{align*}
where we apply Fubini's theorem. We introduce the notation
\begin{align*}
\Lambda_{t,u}^{C}:=\mathbb{E}_{\mathbb{P}}\left[e^{-\int_t^u\lambda^C_s+\lambda^B_s ds}\lambda^C_u \Big \vert \mathcal{F}^{W^f}_T\otimes\mathcal{F}^{W^\perp}_t\right],  \quad u \in [t,T],
\end{align*}
to represent this term that can be computed based on historical information regarding the defaults and then switch to the risk-neutral measure to finally obtain
\begin{align*}
\textnormal{CVA}_t=  \textbf{1}_{\lbrace \tau > t \rbrace}B_t^r\mathbb{E}_{\mathbb{Q}}\left[  (1-R^C)\int_t^T\frac{1}{B_{u}^r}\Lambda_{t,u}^{C}(Q_{u}-C_{u})^+du\Big \vert \mathcal{F}_t\right].  
\end{align*}

This allows the bank to calculate in a first step $\Lambda_{t,u}^{C}$ conditional on the market-factor information and on $t$-conditional historical information from the internal rating system, as in the outlined setting the bank cannot fall back on the dynamics of the intensity rates under $\mathbb{Q}$. In a second step, the usual approximation of the CVA can be performed either via American Monte Carlo or more advanced deep learning methods. In particular, the local risk-minimization approach of the present paper provides a sound mathematical foundation to this appealing computational strategy. Similar calculations hold also for the remaining valuation adjustments.

\section{Capital Value Adjustment (KVA)}\label{Sec:KVA}

In recent years, there has been a debate regarding the need to take into account the funding cost generated by capital requirements on derivative transactions. According to Basel III a bank that is holding a derivative portfolio must set aside capital that can be determined according to different approaches: this could involve a standard model defined by the regulator or an internal model. A first approach that includes the cost of capital in the valuation of derivatives can be found in \cite{gk2014}. Later studies include, among many \cite{Kjaer2017} and \cite{Kjaer2018}, who consider the balance sheet model of \cite{ads2019} or \cite{Castagna2019}. There are other approaches in the literature that see capital requirements as a consequence of market incompleteness, see \cite{Brigo2017}, \cite{alcacre2017}, \cite{Arnsdorf2020}, \cite{css2020} and \cite{crep2022}. 
The capital is provided by the shareholders, who require in return a remuneration for the regulatory capital that they are providing, which is typically termed \textit{hurdle rate}. In line with other value adjustments, the discussion has been about the possibility to charge such costs to the client by means of a suitable value adjustment. In the following, we provide a tentative definition of a capital valuation adjustment that arises from the approach presented in this work. %We assume that any further cost of the strategy beyond $\mathcal{C}_0^{\varphi,A^{R,C}}$ is provided by/provided to the shareholders of the bank. 
We first consider the variation of the cost process of the local risk-minimizing strategy given by
\begin{align*}
    \Delta  \mathcal{C}_t^{\varphi,A^{R,C}}:= \mathcal{C}_t^{\varphi,A^{R,C}}- \mathcal{C}_0^{\varphi,A^{R,C}}.
\end{align*}
This process measures the variation of the cash capital that needs to be put aside or withdrawn. We assume that the capital requirement at a certain point in time $t$ is given by the expected shortfall with confidence level $\alpha$ given by
\begin{align*}
    \mathbb{ES}^\mathbb{P}_\alpha\left[\left.\mathcal{C}_{t+\Delta t}^{\varphi,A^{R,C}}-\mathcal{C}_t^{\varphi,A^{R,C}}\right|\mathcal{G}_t\right]
\end{align*}
 with $\Delta t$ a fixed time horizon, e.g. 1 year. When the shareholders provide such capital, they are remunerated by means of a hurdle rate $h$. When applied to the whole lifetime of the transaction, this gives rise to the following definition

\begin{align*}
    \textnormal{KVA}_t:=B^r_t\mathbb{E}_{\mathbb{Q}}\left[\left.\int_t^{\tau\wedge T}\frac{1}{B^r_u}h_u\mathbb{ES}^\mathbb{P}_\alpha\left[\left.\mathcal{C}_{u+\Delta t}^{\varphi,A^{R,C}}-\mathcal{C}_u^{\varphi,A^{R,C}}\right|\mathcal{G}_u\right]du\right|\mathcal{G}_t\right].
\end{align*}
With this definition the KVA is then charged as a mark-up to the customer on top of the price represented by ${V}_t(\varphi, {A}^{R,C})$. In case the market is complete, then we clearly have that the cost process is a.s. constant, $\mathcal{C}_t^{\varphi,A^{R,C}}-\mathcal{C}_0^{\varphi,A^{R,C}}\equiv 0$ $d\mathbb{P}\otimes dt$-a.s., so that $\textnormal{KVA}_t\equiv 0$ and the client is charged the price ${V}_t(\varphi, {A}^{R,C})$. In summary, the bank charges the amount $V^{full}_t(\varphi, {A}^{R,C})$ given by
\begin{align*}
    V^{full}_t(\varphi, {A}^{R,C}):=V_t(\varphi, {A}^{R,C})+\textnormal{KVA}_t.
\end{align*}
This market practice in the banking sector essentially mimics the traditional insurance valuation, where pricing is performed under the historical measure with the addition of a risk-charge, whose role is taken by the KVA. Clearly, our valuation is performed under a risk neutral setting since we assumed that market risk factors can be hedged. As already claimed in \cite{crep2022} the idea of adding a KVA mimics the valuation approach of Solvency II where the cost of capital is added to the best estimate of the liability.

\begin{remark}In the introduction, we mentioned that all xVA metrics in the present paper are computed at the level of the portfolio of claims between the bank and the counterparty, under the standing principle of \textit{bilateral contestability}. We notice that capital requirements, under a bilateral contestable framework are more conservative, provided that these are defined through a coherent risk measure. In fact, let $\rho$ be a coherent risk measure, and let $(N_i)_{i}$ denote the exposures associated with different netting sets. By subadditivity, one has
\[
\rho\left(\sum_i N_i\right) \leq \sum_i \rho(N_i).
\]
Thus, computing capital requirements at the level of individual netting sets and summing the results yields an upper bound on the portfolio-level capital.

These observations show that restricting attention to bilateral netting sets does not lead to an underestimation of risk. On the contrary, it produces a robust upper bound on capital adjustments by deliberately ignoring diversification effects that are not bilaterally verifiable. In this sense, the bilateral framework can be interpreted as a prudent approximation under information constraints.
\end{remark}

\newpage
\appendix
\section{Föllmer-Schweizer decomposition} \label{appendix}

We briefly recall the definition of the structure condition for a process $Y$ on a filtered probability space $(\Omega, \mathcal{G}_T,\mathbb{G},\mathbb{P})$, see \cite{schweizer_local_risk-minimization}.
\begin{definition} \label{def:MeanVariance}
Let $Y=(Y_t)_{t \in [0,T]}$ be an $\mathbb{R}^d$-valued continuous $(\mathbb{G}, \mathbb{P})$-semimartingale with decomposition $Y=Y_0+A^Y+M^Y.$ We say that $Y$ satisfies the \emph{structure condition (SC)} if
\begin{enumerate}
    \item $Y \in \mathcal{S}^2_{\textnormal{loc,d}}(\mathbb{G},\mathbb{P}),$ 
    \item $A^Y$ is absolutely continuous w.r.t. the bracket process $\langle M^Y \rangle=(\langle M^{Y^i}, M^{Y^j} \rangle)_{i,j=1,...,d}$, i.e. there exists a $\mathbb{G}$-predictable process $\lambda=(\lambda_t)_{ t \in[0, T]}$ such that  $A_t^{Y^i}=\sum_{j=0}^d \int_0^t {\lambda}_u^j d\left\langle M^{Y^i},M^{Y^j}\right\rangle_u$, for $i= 1, \ldots, d $ and $t \in[0, T],$
    \item the \emph{mean-variance tradeoff} process $K^Y=(K_t^Y)_{t\in [0,T]}$ 
    \begin{align} \label{eq:MeanVarianceTradoff}
        {K}_t^Y:=\int_0^t {\lambda}_u^{\top} {d}\left\langle M^Y\right\rangle_u {\lambda}_u= \sum_{i,j=1}^d \int_0^t \lambda_u^i \lambda_u^j d \langle M^{Y^i}, M^{Y^i}\rangle_u
    \end{align}
    satisfies $K_T^Y<\infty$ $\mathbb{P}$-a.s.
\end{enumerate}
\end{definition}
%\begin{remark}
%Let $H\in L^2(\mathbb{G}, \mathbb{P})$. Then by Theorem 1.6 in \cite{schweizer_local_risk-minimization} we have that for a process $Y$ as in Definition \ref{def:MeanVariance} satisfying (SC) such that the mean-variance tradeoff process $K^Y$ is continuous, the definitions of a pseudo-locally risk-minimizing $L^2$-strategy and of a locally risk-minimizing $L^2$-strategy are equivalent, see Definition 1.5 in \cite{schweizer_survey}.
%\end{remark}
It is well-known that there exists a link between no arbitrage conditions and the structure condition for continuous semimartingales. Denote by $\mathcal{P}^{\textnormal{loc}}_{\mathbb{G}}(Y)$ the set of all equivalent local $\mathbb{G}$-martingale measures $\mathbb{Q}$ of $Y$ w.r.t. $\mathbb{P}$.
Then the following result holds, see Theorem 3.5 in \cite{schweizer_survey} and Theorem 3.4 in \cite{Monat_Stricker}.
\begin{theorem} \label{theorem:Schweizer}
    Let $Y$ be a continuous process in $ \mathcal{S}^2_{\textnormal{loc,d}}(\mathbb{G},\mathbb{P})$ and $N$ a random variable in $L^2(\mathcal{G}_T,\mathbb{P})$.
    \begin{enumerate}
     \item If $\mathcal{P}^{\textnormal{loc}}_{\mathbb{G}}(Y) \neq \emptyset,$ then $Y$ satisfies (SC).
     \item If $Y$ satisfies (SC) and the mean-variance tradeoff process $K^Y$ is uniformly bounded, then $N$ admits a unique $\mathbb{G}$-F\"ollmer-Schweizer decomposition.
    \end{enumerate}
\end{theorem}

\bibliographystyle{agsm}
\bibliography{biblio_clean_referenced_only.bib}

\end{document}